\PassOptionsToPackage{final}{graphicx}
\PassOptionsToPackage{nosumlimits,nonamelimits}{amsmath}
\PassOptionsToPackage{colorlinks,linkcolor={blue},citecolor={blue},urlcolor={red},breaklinks=true,final}{hyperref}

\documentclass[a4paper,UKenglish,cleveref, autoref, thm-restate, numberwithinsect,final]{lipics-v2021}
\nolinenumbers

\hideLIPIcs

\usepackage{booktabs}   %
\usepackage{subcaption} %

\providecommand{\comma}{,\operatorname{}\linebreak[1]}           %

\usepackage[utf8]{inputenc}

\usepackage{hypcap}
\usepackage{ifdraft}
\ifdraft{
  \usepackage[draft]{showlabels}
  \usepackage[layout=footnote,marginclue,draft]{fixme}
  \FXRegisterAuthor{sg}{asg}{SG}	%
  \FXRegisterAuthor{pn}{apn}{PN}	%
  \FXRegisterAuthor{ls}{als}{LS}	%
  \FXRegisterAuthor{pw}{apw}{PW}	%
  
  }{
  \usepackage[final]{showlabels}
  
  \usepackage[layout=footnote,final]{fixme}
  \FXRegisterAuthor{sg}{asg}{SG}	%
  \FXRegisterAuthor{pn}{apn}{PN}	%
  \FXRegisterAuthor{ls}{als}{LS}	%
  \FXRegisterAuthor{pw}{apw}{PW}	%
}

\usepackage{cancel}
\usepackage{bbm}
\usepackage{dsfont}
\usepackage{mathtools}
\usepackage[utf8]{inputenc}

\usepackage{tikz-cd}
\usetikzlibrary{decorations.markings}
\usetikzlibrary{arrows.meta}
\usepackage{pedromacros-clean}
\usepackage{centernot}

\tikzcdset{
  only tips/.style={/pgf/tips=true, /tikz/draw=none},
      Corner/.tip={Straight Barb[angle'=90, round, length=4.8pt,width=10pt]},
      Corner Tail/.tip={Corner[reversed, sep=+3pt +1.5]},
         rightcorner/.style={only tips, /tikz/arrows=-Corner},
          leftcorner/.style={only tips, /tikz/arrows=Corner Tail-},
     leftrightcorner/.style={only tips, /tikz/arrows=Corner-Corner},
}

\tikzstyle{shiftarr}=[
rounded corners,%
to path={--([#1]\tikztostart.center)
  -- ([#1]\tikztotarget.center) \tikztonodes
  -- (\tikztotarget)},
]

\usetikzlibrary{
  fit,%
  calc,%
  arrows,%
  arrows.meta,%
  intersections,%
  shapes.misc,%
  shapes.arrows,%
  patterns,%
  automata,%
  chains,%
  matrix,%
  positioning,%
  scopes,%
  decorations.markings,%
  decorations.pathmorphing,%
  external,
  backgrounds
}

\tikzset{
  commutative diagrams/.cd,
  arrow style=tikz,
  diagrams={>={Straight Barb[length=1.75pt,width=3.85pt,inset=1.95pt]}}, %
  row sep=large,
  column sep = huge
}

\tikzset{cong/.style={draw=none,edge node={node [sloped, allow upside down, auto=false]{$\cong$}}},
  iso/.style={draw=none,every to/.append style={edge node={node [sloped, allow upside down, auto=false]{$\cong$}}}}}

\usetikzlibrary{decorations.pathmorphing}

\tikzcdset{scale cd/.style={every label/.append style={scale=#1},
    cells={nodes={scale=#1}}}}

\usepackage{stackengine}
\stackMath
\newcommand\tsup[2][2]{%
  \def\useanchorwidth{T}%
  \ifnum#1>1%
    \stackon[-1.05ex]{\tsup[\numexpr#1-1\relax]{#2}}{\scalebox{2}[1]{$\mathchar"307E$}\kern-.5pt}%
  \else%
    \stackon[-.9ex]{#2}{\scalebox{2}[1]{$\mathchar"307E$}\kern-.5pt}%
  \fi%
}

\usepackage{caption}
\usepackage{subcaption}

\usepackage{xspace}

\usepackage{etoolbox} %

\providecommand{\oname}[1]{{\operatorname{\mathsf{#1}}}}

\newcommand{\nat}{\mathbb{N}}

\newcommand{\dom}{\mathop{\oname{dom}}}
\newcommand{\cod}{\mathop{\oname{cod}}}

\newcommand{\domr}[1]{\lceil\dom #1 \rceil} %
\newcommand{\codr}[1]{\lceil\cod #1 \rceil} %

\providecommand{\ito}{\rightarrowtail}                                       %

\providecommand{\xto}[1]{\,\xrightarrow{#1}\,}
\providecommand{\xtolong}[1]{\,\xlongrightarrow{#1}\,}

\renewcommand{\xto}[1]{\mathrel{\raisebox{-1.15pt}{$\xrightarrow{\hspace{2.5pt}\smash{\raisebox{-2.5pt}{\makebox(3,0)[b]{\scriptsize $#1$}}\hspace{2.5pt}}}$}}}

\providecommand{\xfrom}[1]{\,\xleftarrow{\;#1}\,}
\providecommand{\xfromlong}[1]{\,\xleftarrow{\;#1}\,}

\renewcommand{\xfrom}[1]{\mathrel{\raisebox{-1.15pt}{$\xleftarrow{\smash{\hspace{3.5pt}\raisebox{-2.5pt}{\makebox(3,0)[b]{\scriptsize $#1$}}\hspace{1.5pt}}}$}}}
\providecommand{\mone}{{\text{\kern.5pt\rmfamily-}\mathsf{\kern-.5pt1}}}

\ifdraft{}
{
   \renewcommand{\todo}[1]{}
}

\usepackage{amsthm}

\EventEditors{Olaf Beyersdorff, Micha\l{} Pilipczuk, Elaine Pimentel, and Nguyen Kim Thang}
\EventNoEds{4}
\EventLongTitle{42nd International Symposium on Theoretical Aspects of Computer Science (STACS 2025)}
\EventShortTitle{STACS 2025}
\EventAcronym{STACS}
\EventYear{2025}
\EventDate{March 4--7, 2025}
\EventLocation{Jena, Germany}
\EventLogo{}
\SeriesVolume{327}
\ArticleNo{71}

\title{Identity-Preserving Lax Extensions\\ and Where to Find Them}
\titlerunning{Identity-Preserving Lax Extensions and Where to Find Them}

\author{Sergey Goncharov}{University of Birmingham, UK}{s.goncharov@bham.ac.uk}{https://orcid.org/0000-0001-6924-8766}{Funded by the Deutsche Forschungsgemeinschaft (DFG, German Research Foundation) -- project number 501369690}%

\author{Dirk Hofmann}{CIDMA, University of Aveiro, Portugal}{dirk@ua.pt}{https://orcid.org/0000-0002-1082-6135}{This work is supported by CIDMA under the FCT (Portuguese Foundation for Science and Technology)
Multi-Annual Financing Program for R\&D Units.}

\author{Pedro Nora}{Radboud Universiteit,
  Netherlands}{pedro.nora@ru.nl}{https://orcid.org/0000-0001-8581-0675}{} %

\author{Lutz Schr{\"o}der}{Friedrich-Alexander-Universität Erlangen-Nürnberg,
  Germany}{lutz.schroeder@fau.de}{https://orcid.org/0000-0002-3146-5906}{Funded by the Deutsche Forschungsgemeinschaft (DFG, German Research Foundation) -- project number 531706730}%

\author{Paul Wild}{Friedrich-Alexander-Universität Erlangen-Nürnberg,
  Germany}{paul.wild@fau.de}{https://orcid.org/0000-0001-9796-9675}{Funded by the Deutsche Forschungsgemeinschaft (DFG, German Research Foundation) -- project number 434050016}%

\authorrunning{S. Goncharov, D. Hofmann, P. Nora, L. Schr{\"o}der, P. Wild} %



\Copyright{Sergey Goncharov, Dirk Hofmann, Pedro Nora, Lutz Schröder, and Paul Wild}
\begin{document}

\maketitle

\begin{abstract}
  Generic notions of bisimulation for various types of systems
  (nondeterministic, probabilistic, weighted etc.) rely on
  identity-preserving (\emph{normal}) lax extensions of the functor
  encapsulating the system type, in the paradigm of universal
  coalgebra. It is known that preservation of weak pullbacks is a
  sufficient condition for a functor to admit a normal lax extension
  (the Barr extension, which in fact is then even strict); in the
  converse direction, nothing is currently known about necessary
  (weak) pullback preservation conditions for the existence of normal
  lax extensions. In the present work, we narrow this gap by showing
  on the one hand that functors admitting a normal lax extension
  preserve \emph{1/4-iso} pullbacks, i.e.\ pullbacks in which
  at least one of the projections is an isomorphism.  On the other
  hand, we give sufficient conditions, showing that a functor admits a
  normal lax extension if it weakly preserves either 1/4-iso pullbacks
  and 4/4-epi pullbacks (i.e.\ pullbacks in which all morphisms are
  epic) or inverse images.
   We apply these criteria to concrete examples, in
  particular to functors modelling neighbourhood systems and weighted
  systems.
\end{abstract}

\newpage

\section{Introduction}

Branching-time notions of behavioural equivalence of reactive systems are
typically cast as notions of \emph{bisimilarity}, which in turn are based on
notions of \emph{bisimulation}, the paradigmatic example being Park-Milner
bisimilarity on labelled transition systems~\cite{Milner89}. A key point about
this setup is that while bisimilarity is an equivalence on states, individual
bisimulations can be much smaller than the full bisimilarity relation, and in
particular need not themselves be equivalence relations. In a perspective where
one views bisimulations as certificates for bisimilarity, this feature enables
smaller certificates.

The concept of bisimilarity via bisimulations can be transferred to
many system types beyond basic labelled transition systems, such as
monotone neighbourhood systems~\cite{HansenKupke04}, probabilistic
transition systems, or weighted transition systems. In fact, such
systems can be treated uniformly within the framework of universal
coalgebra~\cite{Rutten00}, in which the system type is encapsulated in
the choice of a set functor (the powerset functor for
non-deterministic branching, the distribution functor for
probabilistic branching etc.). Coalgebraic notions of bisimulation
were originally limited to functors that preserve weak
pullbacks~\cite{Rutten00}, equivalently admit a strictly functorial
extension to the category of relations~\cite{Barr70,Trnkova80}. They
were later extended to functors admitting an
\emph{identity-preserving} or \emph{normal lax
  extension}~\cite{MartiVenema12,MV15} to the category of relations
(this is essentially equivalent to notions of bisimilarity based on
modal logic~\cite{GorinSchroder13}). While there is currently no
formal general definition of what a notion of bisimulation constitutes except
via normal lax extensions, there is a reasonable
claim~\cite{MartiVenema12,MV15} that notions of bisimulation in the
proper sense, in particular with bisimulations not required to be
equivalence relations but stable under key operations such as
relational composition, will not go beyond functors admitting a normal
lax extension.

The \emph{Barr extension} that underlies the original notion of
coalgebraic bisimulation for weak-pullback-preserving
functors~\cite{Rutten00} is, in particular, a normal lax extension;
that is, preservation of weak pullbacks is sufficient for existence of
a normal lax extension. However, this condition is far from being
necessary; there are numerous functors that fail to preserve weak
pullbacks but do admit a normal lax extension, such as the monotone
neighbourhood functor~\cite{MartiVenema12,MV15}. Using the latter
fact, it has been shown that a finitary functor admits a normal lax
extension if and only if it admits a separating set of finitary
monotone modalities~\cite{MartiVenema12,MV15}, cast as monotone
predicate liftings in the paradigm of coalgebraic
logic~\cite{Pattinson04,Schroder08} (a similar result holds for
unrestricted functors if one considers class-sized collections of
infinitary modalities~\cite{GoncharovEA23}).
The latter condition amounts to existence of an expressive modal logic
that has monotone modalities~\cite{Pattinson04,Sch08}, and as such
admits $\mu$-calculus-style fixpoint extensions~\cite{CirsteaEA11}. In
a nutshell, a system type admits a good notion of bisimulation if and
only if it admits an expressive temporal logic. Both sides of this
equivalence, however, need to be witnessed by the construction of a
fairly complicated object; what is missing is a characterization via
\emph{properties} of the underlying functor, rather than via the
existence of extra structure.

In the present work, we narrow the gap between weak pullback
preservation as a sufficient condition for admitting a normal lax
extension, and no known necessary pullback preservation condition. On
the one hand, we establish a necessary preservation condition, showing
that functors admitting a normal lax extension (weakly) preserve
\emph{1/4-iso pullbacks}, i.e.\ pullbacks in which at least one of the
projections is isomorphic. (We often put `weakly' in brackets because
for many of the pullback types we consider, notably for inverse images
and 1/4-iso pullbacks, weak preservation coincides with
preservation.) This is a quite natural condition: A key role in the
field is played by \emph{difunctional relations}~\cite{Rig48}, which
may be thought of as relations obtained by chopping the domain of
an equivalence in half; for instance, given labelled transition
systems~$X$,~$Y$, the bisimilarity relation from~$X$ to~$Y$ is
difunctional. In a nutshell, we show that a functor preserves 1/4-iso
pullbacks iff it acts in a well-defined and monotone manner on
difunctional relations. %
A first application of this necessary condition is a very quick proof
of the known fact that the neighbourhood functor does not admit a
normal lax extension~\cite{MV15}.

We then go on to establish two separate sets of sufficient conditions:
We show that a functor admits a normal lax extension if it (weakly)
preserves either inverse images or both 1/4-iso pullbacks and
\emph{4/4-epi pullbacks}, i.e.\ pullbacks in which all morphisms are
epi (these are also known as \emph{surjective
  pullbacks}~\cite{SeifanEA17}, and weak preservation of 4/4-epi
pullbacks is equivalent to weak preservation of kernel
pairs~\cite{Gumm20}). These sufficient conditions are technically
substantially more involved. As indicated above, they imply that
finitary functors (weakly) preserving either inverse images or 1/4-iso
pullbacks and 4/4-epi pullbacks admit a separating set of finitary
monotone modalities; this generalizes a previous result showing the
same for functors preserving all weak pullbacks~\cite{KurzLeal12}. We
summarize our main contributions in \autoref{fig:contributions}.

\begin{figure}[ht]
\begin{center}
  \begin{tikzpicture}[ node distance=5mm and 44mm, box/.style = {draw,
      minimum height=12mm, inner xsep=3mm, align=center}, sy+/.style =
    {yshift= 0mm}, sy-/.style = {yshift=-0mm}, every edge quotes/.style
    = {align=center} ]

    \node (n1) [box] {weakly\\ preserves\\ pullbacks}; %
    \node (n2) [box,above right=5ex and 4ex of n1]
    {preserves\\inverse images}; %
    \node (n3) [box,below right = 5ex and -2ex of n1] {weakly preserves
      1/4-iso\\ and 4/4-epi pullbacks}; %
    \node (n4) [box,below right = 5ex and 4ex of n2] {admits\\normal
      lax\\extensions}; %
    \node (n5) [box,right = 22ex of n4] {preserves\\ 1/4-iso\\
      pullbacks};
    \draw[dashed, -{Stealth[length=2.3mm]}] (n1) -- (n2); %
    \draw[dashed, -{Stealth[length=2.3mm]}] (n1) -- (n3); %
    \draw[-{Stealth[length=2.3mm]}] (n2) -- (n4) node [midway, above right=-0.3mm and 2mm] {\autoref{thm:main}}; %
    \draw[-{Stealth[length=2.3mm]}] (n3) -- (n4) node [midway, below right=-0.3mm and 2mm] {\autoref{p:3}}; %
    \draw[-{Stealth[length=2.3mm]}] (n4.10) -- (n5.170) node [midway, above] {\autoref{p:450}}; %
    \draw[-{Stealth[length=2.3mm]}] (n5.190) -- node [strike out,draw,-]{} (n4.-10) node [midway, below] {\raisebox{-2ex}{\autoref{p:920}}}; %
  \end{tikzpicture}
  \caption{Summary of main results. Solid arrows are present contributions, dashed arrows are trivial. All implications indicated by arrows are non-reversible; in particular,~\autoref{p:920} shows this for~\autoref{p:450}.}
  \label{fig:contributions}
\end{center}
\end{figure}
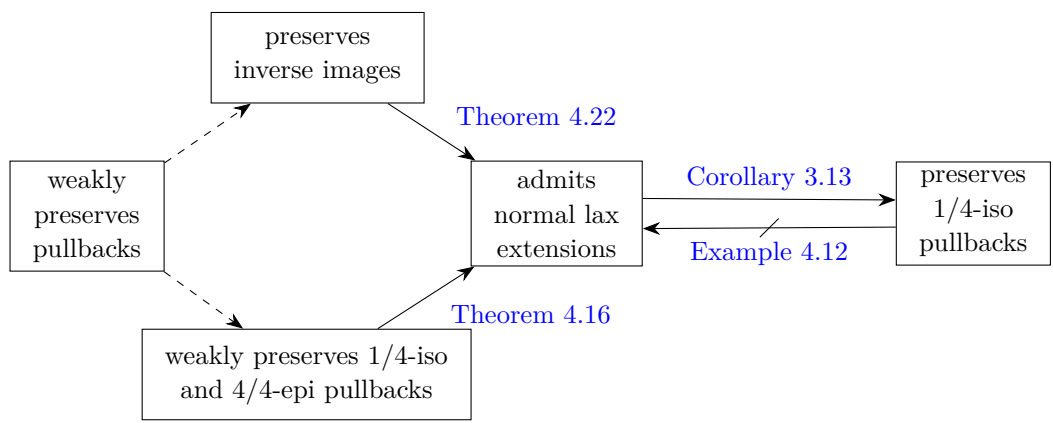
\noindent As per the preceding discussion, these necessary and
sufficient criteria essentially determine (when applicable) whether or
not a given type of systems admits a good notion of bisimulation.

The criterion of weak preservation of 1/4-iso pullbacks and 4/4-epi
pullbacks is satisfied by the monotone neighbourhood functor and
generalizations thereof (e.g.~\cite{SeifanEA17}), and thus in
particular reproves the above-mentioned known fact that functors
admitting separating sets of monotone modalities have normal lax
extensions. The criterion of (weak) preservation of inverse images, in
connection with the necessary criterion, implies that a monoid-valued
functor for a commutative monoid~$M$ (whose coalgebras are
$M$-weighted transition systems) admits a normal lax extension if and
only if~$M$ is positive (which in turn is equivalent to the functor
preserving inverse images~\cite{GummSchroder01}).

\subparagraph*{Related work} With variations in the axiomatics and
terminology, lax extensions go back to an extended strand of work on
relation liftings
(e.g.~\cite{BackhouseEA91,ThijsThesis,HT00,Lev11,Sea05,SS08}). %
We have already mentioned work by Marti and Venema relating lax
extensions to modal logic~\cite{MartiVenema12,MV15}; at the same time,
Marti and Venema prove that the notion of bisimulation induced by a
normal lax extension captures the standard notion of behavioural
equivalence.  \emph{Lax relation liftings}, constructed for functors
carrying a coherent order structure~\cite{HughesJacobs04}, also serve
the study of coalgebraic simulation but obey a different axiomatics
than lax extensions~\cite[Remark~4]{MV15}). Strictly functorial (and
converse-preserving) extensions of set functors to the category of
sets and relations are known to be unique when they exist, and exist
if and only if the functor preserves weak
pullbacks~\cite{Barr93,Trnkova80}; this has been extended to other
base categories~\cite{BackhouseEA91,BirdDeMoor97}. There has been both
longstanding and recent interest in quantitative notions of relation
liftings and lax extensions that act on relations taking values in a
quantale, such as the unit interval, in particular with a view to
obtaining notions of quantitative
bisimulation~\cite{Rutten98,Worrell00,HST14,Gavazzo18,WildSchroder20,WildSchroder22,GoncharovEA23}
that witness low behavioural distance (the latter having first been
treated in coalgebraic generality by Baldan et
al.~\cite{BaldanEA18}). The correspondence between normal lax
extensions and separating sets of modalities generalizes to the
quantitative
setting~\cite{WildSchroder20,WildSchroder22,GoncharovEA23}.

\subparagraph*{Organization} We review material on relations, in particular difunctional relations, and lax
extensions in \autoref{sec:prelims}. In \autoref{sec:f-actions}, we introduce our
necessary pullback preservation condition and show that it characterizes well-definedness
of the natural functor action on difunctional relations. We prove our main
results in \autoref{sec:lax-existence}.
In \autoref{ssec:4/4-epi} we show that a functor that weakly preserves 1/4-iso pullbacks and 4/4-epi pullbacks admits a normal lax extension, and in \autoref{ssec:1/4-mono} we show the same for functors that preserve 1/4-mono pullbacks.

\section{Preliminaries: Relations and Lax Extensions}
\label{sec:prelims}
We work in the category $\SET$ of sets and functions throughout. We
assume basic familiarity with category theory (e.g.~\cite{AHS90}). A
central role in the development is played by (weak) pullbacks: A
commutative square $f\cdot p=g\cdot q$ is a pullback (of $f,g$) if for
every competing square $f\cdot p'=g\cdot q'$, there exists a unique
morphism~$k$ such that $p\cdot k=p'$ and $q\cdot k= q'$; the notion of
weak pullback is defined in the same way except that~$k$ is not
required to be unique. A functor $\ftF$ \emph{weakly preserves} a
given pullback if it maps the pullback to a weak pullback; it is known
that weak preservation of pullbacks of a given type is equivalent to
preservation of weak pullbacks of the same type~\cite[Corollary
4.4]{GS00}. Our interest in functors $\ftF\colon\SET\to\SET$ is driven
mainly by their role as encapsulating types of transition systems in
the paradigm of \emph{universal coalgebra}~\cite{Rutten00}: An
\df{$\ftF$-coalgebra} $(X,\alpha)$ consists of a set~$X$ of
\df{states} and a \df{transition map} $\alpha\colon X\to\ftF X$
assigning to each state $x\in X$ a collection $\alpha(x)$ of
successors, structured according to~$\ftF$. For instance, coalgebras
for the \df{powerset functor~$\calP$} assign to each state a
\emph{set} of successors, and hence are just standard relational
transition systems, while coalgebras for the \df{distribution
  functor\/~$\calD$} (which maps a set~$X$ to the set of discrete
probability distributions on~$X$) assign to each state a distribution
on successor states, and are thus probabilistic transition systems.

A \df{morphism} $f\colon (X,\alpha)\to(Y,\beta)$ of $\ftF$-coalgebras is a map
$f\colon X\to Y$ for which $\beta\cdot f={\ftF f\cdot\alpha}$. Such morphisms are
thought of as preserving the behaviour of states, and correspondingly,
states~$x$ and~$y$ in coalgebras $(X,\alpha)$ and $(Y,\beta)$, respectively,
are \df{behaviourally equivalent} if there exist a coalgebra $(Z,\gamma)$ and
morphisms $f\colon(X,\alpha)\to(Z,\gamma)$, $g\colon(Y,\beta)\to(Z,\gamma)$
such that $f(x)=g(y)$.
\begin{example}
  On relational transition systems, i.e.\ coalgebras for the powerset
  functor~$\calP$, behavioural equivalence instantiates to the usual
  notion of bisimilarity. More generally, labelled transition systems
  with labels taken from a set~$\mathcal A$ are coalgebras for the
  functor $\calP(\mathcal A\times(-))$, and behavioural equivalence
  instantiates to Park-Milner bisimilarity~\cite{AM89}. On Markov
  chains, understood as $\calD$-coalgebras, all states are
  behaviourally equivalent, as all states are identified in the final
  coalgebra $1\to\calD 1\cong 1$. This triviality is removed in
  various forms of probabilistic \emph{labelled} transition systems,
  for instance in $\calD(\mathcal A\times(-))$-coalgebras, on which
  behavioural equivalence instantiates to standard notions of
  probabilistic bisimilarity~\cite{Klin09}.
\end{example}

One is then interested in notions of bisimulation relation that
characterize behavioural equivalence in the sense that two states are
behaviourally equivalent iff they are related by some
bisimulation~\cite{Rutten00,MV15}; this motivates the detailed study
of relations and of extensions of~$\ftF$ that act on relations. We
write $r\colon X \relto Y$ to indicate that~$r$ is a relation from the
set $X$ to the set $Y$ (i.e.\ $r\subseteq X\times Y$), and we write
$x\mathrel{r}y$ when $(x,y)\in r$. Both for functions and for
relations, we use \emph{applicative} composition, i.e.\ given
$r\colon X\relto Y$ and ${s \colon Y\relto Z}$, their composite is
$s\cdot r\colon X\relto Z$ (defined as
$s\cdot r=\{(x,z)\mid\exists y\in
Y.\,x\mathrel{r}y\mathrel{s}z\}$). We say that~$r,s$ of type
$r\colon X\relto Y$ and ${s \colon Y\relto Z}$ are \df{composable},
and we extend this terminology to sequences of relations in the
obvious manner.  Relations between the same sets are ordered by
inclusion, that is \(r \leq r' \iff r \subseteq r'\). We denote by
$1_X\colon X\to X$ the identity map (hence relation) on~$X$, and we
say that a relation \(r \colon X \relto X\) is a \df{subidentity} if
\(r \leq 1_X\). Given a relation $r\colon X \relto Y$,
$r^\circ\colon Y \relto X$ denotes the corresponding converse
relation; in particular, if $f\colon X\to Y$ is a function, then
$f^\circ\colon Y\relto X$ denotes the converse of the corresponding
relation. For a relation $r\colon X\relto Y$, we denote by
$\dom r\subseteq X$ and $\cod r\subseteq Y$ the respective domain and
codomain (i.e.\ $\dom r=\{x\in X\mid\exists y\in Y.\,x\mathrel{r}y\}$
and $\cod r=\{y\in Y\mid\exists x\in X.\,x\mathrel{r}y\}$). A special
class of relations of interest are \df{difunctional relations}
\cite{Rig48}, which are relations factorizable as $g^\circ\cdot f$ for
some functions $f\colon X\to Z$ and $g\colon Y\to Z$, i.e.\ \(x~r~y\)
iff \(f(x) = g(y)\).  In the following we record some folklore facts
about difunctional relations.

\begin{lemma}
	\label{p:48}
	Let \(r \colon X \relto Y\) be a relation.
	Then the following  are equivalent:
	\begin{enumerate}
		\item \label{p:51} \(r\) is difunctional;
		\item \label{p:53} for all \(x_1, x_2\) in \(X\) and~$y_1,y_2\in Y$, if $x_1\mathrel{r} y_1\mathrel{r^\circ} x_2\mathrel{r} y_2$, then $x_1\mathrel{r}y_2$.
		\item \label{p:54} for every span $X\xfrom{\pi_1} R\xto{\pi_2} Y$
		such that \(r = \pi_2 \cdot \pi_1^\circ\), the pushout square
		\begin{center}
			\begin{tikzcd}
				R & Y \\
				X & O
				\ar[from=1-1, to=1-2, "\pi_2"]\
				\ar[from=1-1, to=2-1, "\pi_1"']
				\ar[from=1-1, to=2-2, phantom, very near end, "\ulcorner"]
				\ar[from=2-1, to=2-2, "p_1"']
				\ar[from=1-2, to=2-2, "p_2"]
			\end{tikzcd}
		\end{center}
		is a weak pullback.
	\end{enumerate}
\end{lemma}
As we can see in \autoref{p:48}(\ref{p:54}) above, difunctional relations are characterized as weak pullbacks, and in this regard we recall that generally, a commutative square $f\cdot p=g\cdot q$ is a weak pullback iff $q\cdot p^\circ = g^\circ\cdot f$, equivalently $p\cdot q^\circ=f^\circ\cdot g$.

 The \df{difunctional closure} of a relation
$r\colon X\relto Y$ is the least difunctional relation $\hat r \colon X \relto Y$ greater than or
equal to~$r$.
	\label{lem:df_fp}
	It follows from \autoref{p:48} that the difunctional closure of a relation \(r \colon X \relto Y\) given by a span $X\xfrom{\pi_1} R\xto{\pi_2} Y$
	is obtained by computing its pushout $X\xto{p_1} O\xfrom{p_2} Y$; i.e., the difunctional closure~\(\hat{r}\) of~\(r\) is the relation \(p_2^\circ \cdot p_1\).
	More explicitly, $\hat r = \bigvee_{n\in\nat} r\cdot (r^\circ\cdot r)^n$ (e.g. \cite{Rig48,GZ14}).

A \df{lax extension} $\eF$ of an endofunctor $\ftF\colon\SET \to \SET$ is a
mapping that sends any relation $r \colon X \relto Y$ to a relation $\eF r
	\colon \ftF X \relto \ftF Y$ in such a way that
\vskip1\baselineskip
\begin{conditions}
	\item[\nlabel{p:61}{(L1)}] $r \le r' \implies \eF r \le \eF r'$,
	\item[\nlabel{p:26}{(L2)}] $\eF s\cdot\eF r\le\eF (s\cdot r)$,
	\item[\nlabel{p:0}{(L3)}] $\ftF f \le \eF f$ and ${(\ftF f)}^\circ\le\eF (f^\circ)$,
\end{conditions}
\vskip1\baselineskip\noindent
for all \(r \colon X \relto Y\), \(s \colon Y \relto Z \) and \(f \colon X \to
Y\).
We define \df{relax extensions} in the same way, without however requiring
property \ref{p:26}. We call a (re)lax extension \df{identity-preserving}, or
\df{normal}, if \(\eF 1_X = 1_{\ftF X}\) for every set~\(X\), and we say that a
(re)lax extension \df{preserves converses} if \(\eF(r^\circ) = (\eF r)^\circ\).

A tactical advantage of using the term ``relax extension'' is that we can thus
refer to constructions that produce lax extensions most of the time, except for
some cases when \ref{p:26} may fail. A prototypical example of this sort is the
\df{Barr extension}~$\ftbF$ \cite{Bar70}, which for
weak-pullback-preserving~$\ftF$ is even a strict extension, and is defined as
follows. Given a relation \(r \colon X \relto Y\), choose a factorization
\(\pi_2 \cdot \pi_1^\circ\) for some span $X\xfrom{\pi_1} R\xto{\pi_2} Y$
and put $\ftbF r = \ftF \pi_2 \cdot {(\ftF \pi_1)}^\circ$. This assignment is
independent of the factorization of \(r\), and $r$ admits a \df{canonical factorization} which is given by projecting into $X$ and $Y$ the subset of $X \times Y$ of pairs of elements related by $r$.
It is well-known that for every \(\SET\)-functor, the Barr extension is a normal relax extension, but it is a
lax extension precisely when $\ftF$ preserves weak pullbacks~\cite{KupkeEA12}.
In this case, the Barr extension is also the least lax extension of~$\ftF$, for
it follows from \ref{p:61}--\ref{p:0} that \(\ftF \pi_2 \cdot {(\ftF
	\pi_1)}^\circ\leq \eF r\) for every lax extension $\eF$.

Lax extensions have been used extensively to treat the notion of bisimulation
coalgebraically (e.g. \cite{HT00,Lev11,MV15}).
Given a lax extension \(\eF \colon \REL \to \REL\) of a functor
\(\ftF \colon \SET \to \SET\), an \df{\(\eF\)-simulation} between
\(\ftF\)-coalgebras \((X,\alpha)\) and \((Y,\beta)\) is a relation
\(s \colon X \relto Y\) such that
\(\beta \cdot s \leq \eF s \cdot \alpha\), that is, whenever
$x\mathrel{r}y$, then $\alpha(x)\mathrel{Lr}\beta(y)$. If \(\eF\)
preserves converses, then \(\eF\)-simulations are more suitably called
\df{\(\eF\)-bisimulations}.  Between two given coalgebras, there is a
greatest \(\eF\)-(bi)simulation, which is termed
\df{\(\eF\)-(bi)similarity}. It has been shown \cite{MV15} that
if~$\eF$ is normal and preserves converses, then \(\eF\)-bisimilarity
coincides with coalgebraic behavioural equivalence as recalled
above. The axioms of lax extensions guarantee that $L$-bisimulations
are closed under converse and composition and that coalgebra morphisms
are (functional) $L$-bisimulations, so that $L$-bisimilarity includes
behavioural equivalence; that is, $L$-bisimilarity is \emph{complete}
for behavioural equivalence. Normality of lax extensions ensures that
$L$-bisimulations are \emph{sound} for behavioural equivalence, i.e.\
$L$-bisimilarity is included in behavioural equivalence.
\begin{example}\label{expl:lax}
  \begin{enumerate}
  \item For relational transition systems, understood as
    $\calP$-coalgebras, we have a normal lax extension~$L$ of~$\calP$
    given by the standard Barr extension, which in turn coincides with
    the well-known Egli-Milner extension: Given $r\colon X\relto Y$,
    $S\in\calP X$, and $T\in\calP Y$, we have $S\mathrel{Lr}T$ iff for
    all $x\in S$, there exists $y\in T$ such that $x\mathrel{r}y$, and
    symmetrically. An $L$-bisimulation is then just a bisimulation in
    the standard sense.
  \item\label{item:lax-prob} On $F=\calD(\mathcal A\times(-))$, we
    have a normal lax extension~$L$ given for $r\colon X\relto Y$,
    $\mu\in\calD(\mathcal A\times X)$,
    $\nu\in\calD(\mathcal A\times Y)$ by $\mu\mathrel{Lr}\nu$ iff for
    all $l\in\mathcal{A},A\in\calP X$, we have
    $\nu(\{l\}\times r[A])\ge\mu(\{l\}\times A)$, and
    symmetrically~\cite{GorinSchroder13}. The arising notion of
    $L$-bisimulation is sound and complete for probabilistic
    bisimilarity on probabilistic labelled transition systems.
  \end{enumerate}
\end{example}

\begin{remark}
  As mentioned in the introduction, a functor $\ftF$ admits a normal
  lax extension iff $\ftF$ admits a separating class of monotone
  predicate liftings~\cite{MV15,GoncharovEA23}. For readability, we
  discuss only the case where both the functor and the predicate
  liftings are finitary~\cite{MV15}.  An \emph{$n$-ary predicate
    lifting}~$\lambda$ for~$\ftF$ is a natural transformation of type
  $\lambda\colon\contrapow^n\to\contrapow\cdot\ftF^\op$
  where~$\contrapow$ denotes the contravariant powerset functor (given
  by~$\contrapow X$ being the powerset of a set~$X$, and
  $\contrapow f(B)=f^{-1}[B]$ for $f\colon X\to Y$ and
  $B\in\contrapow Y$); that is, for a set~$X$, $\lambda_X$ lifts~$n$
  predicates on~$X$ to a predicate on $\ftF X$. Predicate liftings
  determine modalities in coalgebraic modal
  logic~\cite{Pattinson04,Sch08}; a basic example is the unary
  predicate lifting~$\lambda$ for the (covariant) powerset
  functor~$\calP$ given by
  $\lambda_X(A)=\{B\in\calP X\mid B\subseteq A\}$ for a
  predicate~$A\subseteq X$, which determines the standard box modality
  on $\calP$-coalgebras, i.e.\ on standard relational transition
  systems. A set of predicate liftings is \emph{separating} if
  distinct elements of~$\ftF X$ can be separated by lifted predicates;
  this condition ensures that the associated instance of coalgebraic
  modal logic is \emph{expressive}, i.e.\ separates behaviourally
  inequivalent
  states~\cite{Pattinson04,Sch08}. %
	Monotonicity of predicate liftings allows the definition of modal fixpoint
	logics for temporal specification~\cite{CirsteaEA11}. In the mentioned
	correspondence between lax extensions and predicate liftings, the construction
	of predicate liftings from a lax extension~$\eF$ roughly speaking involves
	application of~$\eF$ to the elementhood relation.
\end{remark}

\section{Functor Actions on Difunctional Relations}
\label{sec:f-actions}

Our pullback preservation criterion for existence of normal lax extensions
grows from an analysis of how functors act on difunctional relations. To start
off, it is well-known that normal lax extensions of a given \(\SET\)-functor
are given on difunctional relations by the action of the functor
(e.g.~\cite{MV15,HST14}):
\begin{proposition}
	\label{p:215}
	Let\/ \(\eF\) be an assignment of relations\/ $\eF r\colon \ftF X\relto\ftF Y$ to relations $r\colon X\relto Y$ that satisfies \ref{p:61}, \ref{p:26} as well as $1_{\ftF X}\leq \eF 1_X$ for all $X \in \SET$.
	Then\/ \(\eF\) is a lax extension of\/ \(\ftF\) iff for all functions \(f \colon W \to X\), \(g \colon Z \to Y\) and relations \(r \colon X \relto Y\),
	\(
	\eF(g^\circ \cdot r \cdot f) = (\ftF g)^\circ \cdot \eF r \cdot \ftF f.
	\)
\end{proposition}

\begin{corollary}
	\label{p:27}
	All normal lax extensions of a given \(\SET\)-functor coincide on difunctional relations.
	Specifically, for every normal lax extension\/ \(\eF\) of\/ \(\ftF \colon \SET \to \SET\), \(\eF(g^\circ \cdot f) = {\ftF g}^\circ \cdot \ftF f\) for all \(f \colon X \to A\) and \(g \colon Y \to A\).
\end{corollary}
\noindent Therefore, a functor \(\ftF \colon \SET \to \SET\) that
admits at least one normal lax extension must be \df{monotone on
  difunctional relations} in the following sense: for all difunctional
relations \(g^\circ \cdot f \colon X \relto Y\) and
\(g'^\circ \cdot f' \colon X \relto Y\), if
\(g^\circ \cdot f \leq g'^\circ \cdot f'\) then
\((\ftF g)^\circ \cdot \ftF f \leq (\ftF g')^\circ \cdot \ftF f'\).
This property no longer mentions lax extensions, and implies that the
functor is \df{well-defined on difunctional relations}, i.e.\ that
$\ftF$ sends cospans that determine the same difunctional relation to
cospans that determine the same difunctional relation.  In this
section, we show that being monotone on difunctional relations is
equivalent to preserving 1/4-iso (2/4-mono) pullbacks in the sense
defined next; as indicated in the introduction, this allows for a
quick proof of the fact that the neighbourhood functor fails to admit
a normal lax extension~\cite{MV15}. On this occasion, we also discuss
various types of pullbacks and their (weak) preservation in some more
breadth for later use in our sufficient criteria
(\autoref{sec:lax-existence}).

\begin{definition}
	\label{p:18}
	We say that a functor \(\ftF \colon \SET \to \SET\) preserves \df{1/4-iso
		2/4-mono pullbacks}, \df{1/4-iso pullbacks}, \df{1/4-mono pullbacks} and \df{inverse images} if it
	sends pullbacks of the following forms, respectively, to pullbacks, with arrows
	$\rightarrowtail$ and $\xto{\raisebox{0.5ex}{$\simeq$}}$ indicating injectivity and bijectivity
	correspondingly.
	\begin{center}
		\begin{tikzcd} %
			P        \ar[from=1-1, to=2-2, phantom, very near start, "\lrcorner"] & B \\
			X & Y
			\ar[from=1-1, to=1-2, "\simeq"]
			\ar[from=1-1, to=2-1, tail]
			\ar[from=1-2, to=2-2, tail]
			\ar[from=2-1, to=2-2]
		\end{tikzcd}
		\qquad
		\begin{tikzcd} %
			P        \ar[from=1-1, to=2-2, phantom, very near start, "\lrcorner"] & B \\
			X & Y
			\ar[from=1-1, to=1-2, "\simeq"]
			\ar[from=1-1, to=2-1]
			\ar[from=1-2, to=2-2]
			\ar[from=2-1, to=2-2]
		\end{tikzcd}
		\qquad
		\begin{tikzcd} %
			P        \ar[from=1-1, to=2-2, phantom, very near start, "\lrcorner"] & B \\
			X & Y.
			\ar[from=1-1, to=1-2]
			\ar[from=1-1, to=2-1, tail]
			\ar[from=1-2, to=2-2]
			\ar[from=2-1, to=2-2]
		\end{tikzcd}
		\qquad
		\begin{tikzcd} %
			P        \ar[from=1-1, to=2-2, phantom, very near start, "\lrcorner"] & B \\
			X & Y
			\ar[from=1-1, to=1-2]
			\ar[from=1-1, to=2-1, tail]
			\ar[from=1-2, to=2-2, tail]
			\ar[from=2-1, to=2-2]
		\end{tikzcd}
	\end{center}
\end{definition}

\begin{remark}
	 1/4-Iso 2/4-mono pullbacks are special inverse images,
	characterized by the property that the fibre over every element in the image of the function $B \rightarrowtail Y$ is a singleton.
	In particular, the inverse image of the empty subset is a 1/4-iso 2/4-mono pullback.
\end{remark}
\noindent
Due to the following proposition, for consistency, we tend to use
``preservation of 1/4-mono pullbacks'' instead of ``preservation of
inverse images''.

\begin{proposition}
	\label{p:910}
	A $\SET$-functor preserves 1/4-mono pullbacks iff it preserves inverse images.
\end{proposition}

\noindent Similarly, we will see in \autoref{p:64} that preservation
of 1/4-iso pullbacks is equivalent to preservation of 1/4-iso 2/4-mono
pullbacks.  We thus tend to use the terms ``1/4-iso 2/4-mono pullback
preserving'' and ``1/4-iso pullback preserving''
interchangeably. Furthermore, in \autoref{expl:functors} we will see
that preservation of 1/4-mono pullbacks is properly stronger than
preservation of 1/4-iso pullbacks.

	Each of the preservation properties introduced in \autoref{p:18} implies preservation of
	monomorphisms, even if we only require that the corresponding pullbacks are
	weakly preserved. Hence, as at least one of the projections of the pullbacks is
	monic, preserving the pullbacks mentioned is equivalent to weakly preserving
	them, and, therefore, each of the properties is implied by weakly preserving
	pullbacks.
	Also, note that weakly preserving limits of a given shape is equivalent to preserving weak limits of that shape (e.g. \cite[Corollary~4.4]{GS00}).
	Furthermore, weakly preserving pullbacks is known to be sufficient for the existence of a normal lax extension -- the Barr extension -- and  this condition can be decomposed as follows:

\begin{theorem}
	\cite[Theorem 2.7]{GS05}
	A $\SET$-functor weakly preserves pullbacks iff it weakly preserves inverse images and kernel pairs.
\end{theorem}

\noindent It turns out that weakly preserving kernel pairs is
equivalent to weakly preserving 4/4-epi pullbacks as defined next.

\begin{definition}
	We say that a functor \(\ftF \colon \SET \to \SET\) weakly preserves \df{4/4-epi pullbacks}, if it
	sends pullbacks of the form
	\begin{center}
		\begin{tikzcd} %
			P        \ar[from=1-1, to=2-2, phantom, very near start, "\lrcorner"] & B \\
			X & Y,
			\ar[from=1-1, to=1-2, twoheadrightarrow]
			\ar[from=1-1, to=2-1, twoheadrightarrow]
			\ar[from=1-2, to=2-2, twoheadrightarrow]
			\ar[from=2-1, to=2-2, twoheadrightarrow]
		\end{tikzcd}
	\end{center}
	with arrows $\twoheadrightarrow$ indicating surjectivity, to weak pullbacks (necessarily of surjections).
\end{definition}

\begin{theorem}
	\cite[Corollary 5]{Gumm20}
	A $\SET$-functor weakly preserves kernel pairs iff it weakly preserves 4/4-epi pullbacks.
\end{theorem}

\noindent Therefore, the condition of weakly preserving pullbacks can
be decomposed as:

\begin{corollary}
	A $\SET$-functor weakly preserves pullbacks iff it weakly preserves 1/4-mono pullbacks and 4/4-epi pullbacks.
\end{corollary}
In Section~\ref{sec:lax-existence}, we will show that either preserving 1/4-mono pullbacks or weakly preserving 1/4-iso pullbacks and  4/4-epi pullbacks is sufficient for the existence of a normal lax extension.

\begin{example}
  \label{expl:functors}
  \begin{enumerate}[wide]
  \item The subfunctor $(-)^3_2 \colon \SET \to \SET$ of the functor
    $(-)^3 \colon \SET \to \SET$ that sends a set $X$ to the set of
    triples of elements of $X$ consisting of at most two distinct
    elements does not preserve pullbacks weakly \cite{AM89} but it
    preserves inverse images.

  \item \label{p:950} The neighbourhood functor
    \(\calN \colon \SET \to \SET\) (whose coalgebras are neighbourhood
    frames~\cite{Chellas80}) sends a set \(X\) to the set
    \(\calN X = \calP\calP X\) of neighbourhood systems over~\(X\),
    and a function \(f \colon X \to Y\) to the function
    \(\calN f \colon \calN X \to \calN Y\) that assigns to every
    element \(\calA \in \calN X\) the set
    \(\{B \subseteq Y \mid f^{-1}[B] \in \calA\}\).  The monotone
    neighbourhood functor \(\calM \colon \SET \to \SET\) is the
    subfunctor of the neighbourhood functor that sends a set \(X\) to
    the set of upward-closed subsets of \((\calP X, \subseteq)\). Its
    coalgebras are monotone neighbourhood frames, which feature, e.g.,
    in the semantics of game logic~\cite{Parikh83} and concurrent
    dynamic logic~\cite{Peleg85}.  A closely related functor is the
    clique functor \(\calC \colon \SET \to \SET\), which is the
    subfunctor of~$\calM$ given by
    $\calC X=\{\alpha\in\calM X\mid\forall A,B\in\alpha.\,A\cap
    B\neq\emptyset\}$. The functors \(\calM\) and \(\calC\) do not
    preserve inverse images: Consider the sets \(3 = \{0,1,2\}\) and
    \(2 = \{a,b\}\). Let \(e \colon 3 \to 2\) be the function that
    sends \(0,1\) to \(a\) and \(2\) to \(b\), and \(B = \{a\}\). Then
    \(\calM e(\upc\{0,1\} \cup \upc\{1,2\}) = \upc\{a\}\), where
    $\upc$ denotes upwards closure, but \(e^{-1}[B] = \{0,1\}\) and
    \(\upc\{0,1\} \cup \upc\{1,2\}\) does not belong to the image of
    the function
    $\calM i \colon \calM\{0,1\} \rightarrowtail \calM 3$, where
    $i \colon \{0,1\} \rightarrowtail 3$ denotes the corresponding
    inclusion.  However, routine calculations show that these functors
    do preserve 1/4-iso (2/4-mono) pullbacks and weakly preserve
    4/4-epi pullbacks (for the first functor, see~\cite[Proposition
    4.4]{SeifanEA17}).
  \item \label{p:403} Given a commutative monoid~$(M,+,0)$ (or
    just~$M$), the \emph{monoid-valued functor} $M^{(-)}$ maps a
    set~$X$ to the set $M^{(X)}$ of functions $\mu\colon X\to M$ with
    \emph{finite support}, i.e.\ $\mu(x)\neq 0$ for only finitely
    many~$x$. The coalgebras of $M^{(-)}$ are $M$-weighted transition
    systems. It is known that $M^{(-)}$ preserves inverse images
    iff~$M$ is \emph{positive}, i.e.\ does not have non-zero
    invertible elements~\cite[Theorem 5.13]{GummSchroder01} (the cited
    theorem shows the equivalence for non-empty inverse images; it is
    easy to check that in case~$M$ is positive, $M^{(-)}$ preserves
    empty pullbacks). Moreover, $M^{(-)}$ preserves weak
    pullbacks iff~$M$ is positive and \emph{refinable}, i.e.\ whenever
    $m_1+m_2=n_1+n_2$ for $m_1,m_2,n_1,n_2\in M$, then there exists a
    $2\times 2$-matrix with entries in~$M$ whose $i$-th column sums up
    to~$m_i$ and whose $j$-th row sums up to~$n_j$, for
    $i,j\in\{1,2\}$~\cite[Theorem 5.13]{GummSchroder01}.  Monoids that
    are positive but not refinable are fairly
    common~\cite{ClementinoEA14}; the simplest example is the additive
    monoid $\{0,1,2\}$ where $2+1=2$.
    
    The functor \(M^{(-)}\) preserves 1/4-iso (2/4-mono) pullbacks iff
    it preserves inverse images iff \(M\) is positive. Indeed, suppose
    that \(M\) is not positive. Consider the functions
    \(!_2 \colon 2 \to 1\),
    \(!_\varnothing \colon \varnothing \to 1\). Then, for mutually
    inverse non-zero elements \(u\) and \(v\) of \(M\), the function
    \(M^{(!_2)}\) sends both the pair \((0,0)\) and the pair \((u,v)\) to
    \(0 \in M^{(1)}\), which is in the image of
    $M^{(!_\varnothing)} \colon M^{(\varnothing)} \to M^{(1)}$.
    Therefore, the functor \(M^{(-)}\) does not preserve the (1/4-iso)
    pullback of $(!_2, !_\varnothing)$: This pullback has
    vertex~$\varnothing$, and $M^{(\varnothing)}$ has only one
    element.

  \item %
    In recent work~\cite{Gumm20}, it has been shown that the functor
    of a monad induced by a variety of algebras preserves inverse
    images iff whenever a variable~$x$ is canceled from a term when
    identified with other variables, then the term does not actually
    depend on~$x$. This provides a large reservoir of functors that
    preserve inverse images but do not always have easily guessable
    normal lax extensions (whose existence will however be guaranteed
    by our main results).  One example is the functor that maps a
    set~$X$ to the free semigroup over~$X$ quotiented by the equation
    $x x x = x x$, as neither this equation nor associativity cancel
    any variables.  Notice that this functor does not preserve 4/4-epi
    pullbacks.
  \end{enumerate}
\end{example}

Finally, we show that being monotone on difunctional relations is equivalent to
preserving 1/4-iso (2/4-mono) pullbacks.
The next lemma connects the order on difunctional relations and pullbacks of such type.

\begin{lemma}
	\label{p:62}
	Let $X\xto{f} A\xfrom{g} Y$ and $X\xto{f'} A'\xfrom{g'} Y$ be cospans for which
	there is a map \(h \colon A \to A'\) such that \(f' = h \cdot f\) and \(g' = h
	\cdot g\). Moreover, consider the commutative square
	\begin{equation}\label{d:eq:1}
		\begin{tikzcd}
			f\left[X\right]\cap g\left[Y\right]\ar{r}{h'}\ar[>->]{d}
			& f'\left[X\right]\cap g'\left[Y\right]\ar[>->]{d}\\
			A\ar{r}{h} & A'
		\end{tikzcd}
	\end{equation}
	where \(h' \colon f[X] \cap g[Y] \to f'[X] \cap g'[Y]\) is the restriction of~\(h\) to \(f[X] \cap g[Y]\) and the vertical arrows denote subset inclusions.
	\begin{enumerate}
		\item If \(g^\circ \cdot f \geq g'^\circ \cdot f'\), then \(h'\) is a bijection.
		\item If \(g^\circ \cdot f \geq g'^\circ \cdot f'\) and the cospan \((f,g)\) is epi,
		      then \eqref{d:eq:1} is a pullback.
		\item\label{item:difun-incl} If \(h'\) is a bijection and \eqref{d:eq:1} is a pullback,
		      then \(g^\circ \cdot f \geq g'^\circ \cdot f'\).
	\end{enumerate}
\end{lemma}
\noindent (Notice in particular that if~$h'$ is a bijection
and~\eqref{d:eq:1} is a pullback, then~\eqref{d:eq:1} is a 1/4-iso
pullback.) Using \autoref{p:62}, one proves the announced
characterization:
\begin{theorem}
  \label{p:64}
  The following clauses are equivalent for a functor \(\ftF \colon \SET \to \SET\):
  \begin{enumerate}%
  \item \label{p:49} \(\ftF\) preserves 1/4-iso 2/4-mono pullbacks.
  \item \label{p:52} \(\ftF\) is well-defined on difunctional relations.
  \item \label{p:110} \(\ftF\) is monotone on difunctional relations.
  \item \label{p:10} \(\ftF\) preserves 1/4-iso pullbacks.
  \end{enumerate}
\end{theorem}

\begin{corollary}
	\label{p:450}
	If a $\SET$-functor admits a normal lax extension, then it preserves 1/4-iso pullbacks.
\end{corollary}

\noindent Therefore, the following functors do not admit a normal lax
extension.

\begin{example}
  \label{p:860}
  \begin{enumerate}[wide]
  \item The neighbourhood functor $\calN \colon \SET \to \SET$ (cf.\
    \autoref{expl:functors}(\ref{p:950})) does not preserve 1/4-iso
    pullbacks: the element \(\calP 1 \in \calN 1\) belongs to the
    image of the function $\calN !_\varnothing$, with
    $!_\varnothing \colon \varnothing \rightarrowtail 1$, however, its
    fiber w.r.t.  \(\calN !_2\), with \(!_2 \colon 2 \to 1\), is not a
    singleton.
  \item For every non-positive commutative monoid, the monoid valued
    functor $M^{(-)} \colon \SET \to \SET$ does not preserve 1/4-iso
    pullbacks (\autoref{expl:functors}(\ref{p:403})).
  \item More generally, by (the proof of)
    \cite[Proposition~4.4]{ClementinoEA14}, the functor
    \(\ftF \colon\SET\to\SET\) of the monad induced by a variety of
    algebras that admits a weak form of subtraction (for instance,
    groups, rings,  vector spaces) does not preserve 1/4-iso
    pullbacks.
  \item \label{p:861} For every set \(A\) with at least two elements,
    consider the functor \(\SET(A,-)/\sim\) that maps a set \(X\) to
    the quotient of the set \(\SET(A,X)\) by the equivalence relation
    \(\sim\) that identifies exactly all non-injective
    maps, %
    and maps a function \(f \colon X \to Y\) to the one sending the
    equivalence class of $g\colon A\to X$ to that of $f\cdot g$. This
    functor does not preserve 1/4-iso pullbacks. For instance, for
    \(A = \{0,1\}\), consider the sets \(3 = \{a,b,c\}\) and
    \(B = \{0\}\). Then, the fibre of each element of
    \(B \subseteq A\) w.r.t. the function \(f \colon 3 \to A\) that
    sends \(a\) to \(0\) and \(b,c\) to \(1\) is a singleton; however,
    the fibre of the equivalence class of the constant map into \(0\)
    w.r.t. \(\SET(A,f)_\sim\) is not a singleton. Similar
    counterexamples can be constructed for arbitrary \(A\) with at
    least two elements.
  \end{enumerate}
\end{example}

\begin{remark}
  Every coalgebra can be quotiented by behavioural equivalence
  (e.g.~\cite{Gumm03}). Such a quotient can be described by a
  \emph{cocongruence} on a given coalgebra, i.e.\ an equivalence
  relation that is compatible with the coalgebra structure, and, of
  course, a cocongruence can be specified by a generating relation
  that need not itself be an equivalence. For instance, cocongruences
  have been studied in the context of linear weighted
  automata~\cite{BonchiEA12} (where they are in fact termed
  bisimulations), and even on neighbourhood frames, one obtains such a
  notion of equivalence-witnessing relation from the standard Barr
  extension~\cite{MV15}. All this does not contradict the moral claim
  that, by \autoref{p:860}, there are no `good' notions of
  bisimulation for, e.g., neighbourhood frames or integer-weighted
  transition systems, as (generating relations of) cocongruences are
  missing some of the features that we include in the wish list for
  bisimulations and that $L$-bisimulations do provide (cf.\
  \autoref{sec:prelims}). Notably, cocongruences work only on a single
  coalgebra (while we expect bisimulations to connect two possibly
  different coalgebras), and they fail to be closed under relational
  composition.
\end{remark}

\section{Existence of Normal Lax Extensions}
\label{sec:lax-existence}

We proceed to present the main results of the paper: a \(\SET\)-functor that
weakly preserves 1/4-iso pullbacks and 4/4-epi pullbacks, or that preserves
1/4-mono pullbacks admits a normal lax extension. In view of the facts recalled
in \autoref{sec:prelims}, this means that these functors admit a notion of
bisimulation that captures behavioural equivalence, or equivalently, that they
admit a separating class of monotone predicate liftings.

We begin by showing that the smallest lax extension of a \(\SET\)-functor
is obtained by ``closing its Barr relax extension under composition''. As a
consequence, in \autoref{cor:least_flat} we obtain a criterion to
determine if a \(\SET\)-functor admits a normal lax extension.

Consider the partially ordered classes \(\LaxF\) and \(\ReLaxF\) of lax and
relax extensions of~$\ftF$, respectively, ordered pointwise.
With the following result we can construct lax extensions from relax extensions in a universal way.

\begin{proposition}
	\label{p:34}
	Let \(\ftF \colon \SET \to \SET\) be a functor.
	The inclusion \(\LaxF \ito \ReLaxF\) has a left adjoint \(\laxif{(-)} \colon \ReLaxF \to \LaxF\) that sends a relax extension \(\eR \colon \REL \to \REL\) of \(\ftF\) to its \df{laxification} \(\laxif{\eR} \colon \REL \to \REL\), which is defined on \(r \colon X \relto Y\) by
	\begin{equation}
		\label{p:410}
		\laxif{\eR} r = \bigvee_{\substack{r_1, \ldots, r_n \colon \\ r_n \cdot \ldots \cdot r_1 \leq r}} \eR r_n \cdot \ldots \cdot \eR r_1.
	\end{equation}
	Furthermore, if a relax extension \(\eR \colon \SET \to \SET\) preserves
	converses, then so does its laxification.
\end{proposition}
Since every lax extension of a functor is greater than or equal to the Barr relax
extension (cf.\ \autoref{sec:prelims}), we thus have:

\begin{corollary}
	\label{p:42}
	The smallest lax extension of a functor is given by the laxification of its Barr relax extension.
\end{corollary}
For the Barr relax extension of a $\SET$-functor, the supremum in the formula~(\ref{p:410}) can be restricted as follows.

\begin{lemma}
	\label{p:605}
	For every composable sequence $r_1, \ldots, r_n$ such that $r_n \cdot \ldots \cdot r_1 \leq r$, for some relation $r$, there is a composable sequence $r_1', \ldots, r_n'$ such that $r_n' \cdot \ldots \cdot r_1' = r$  and $\ftbF r_n \cdot \ldots \cdot \ftbF r_1 \leq \ftbF r_n' \cdot \ldots \cdot \ftbF r_1'$.
\end{lemma}

\begin{corollary}
	Let $\ftF \colon \SET \to \SET$ be a functor.
	For every relation $r \colon X \relto Y$,
	\begin{align*}
		\laxif{(\ftbF)} r = \bigvee_{\substack{r_1, \ldots, r_n \colon \\ r_n \cdot \ldots \cdot r_1 = r}} \ftbF r_n \cdot \ldots \cdot \ftbF r_1.
	\end{align*}
\end{corollary}
Therefore, as normality of a lax extension also implies normality of any lax extension
below it, we have
\begin{corollary}
	\label{cor:least_flat}
	A functor \(\ftF \colon \SET \to \SET\) admits a normal lax extension iff the laxification of its Barr relax extension is normal.
	More concretely, a functor \(\ftF \colon \SET \to \SET\) admits a normal lax extension iff for every set \(X\) and every composable sequence of relations \(r_1,\ldots,r_n\), whenever \(r_n \cdot \ldots \cdot r_1 = 1_X\), then $\ftbF r_n \cdot \ldots \cdot \ftbF r_1 \leq 1_{\ftF X}$.
\end{corollary}

\begin{remark}
	It is well-known \cite{Bar70} that for every functor $\ftF \colon \SET \to \SET$ and all relations $r \colon X \relto Y$ and $s \colon Y \relto Z$, $\ftbF (s \cdot r) \leq \ftbF s \cdot \ftbF r$.
	Hence, once we show the inequality of \autoref{cor:least_flat} we actually have equality.
\end{remark}

In general terms, our main results follow by showing that in \autoref{cor:least_flat}, under certain conditions on $\SET$-functors, it suffices to consider composable sequences of relations that satisfy nice properties.
In this regard, it is convenient to introduce the following notion.

\begin{definition}
	Let $r_1, \ldots, r_n$ be a composable sequence of relations.
	A composable sequence $s_1, \ldots, s_k$ is said to be a \df{Barr upper bound}
	 of the sequence $r_1, \ldots, r_n$ if $r_n \cdot \ldots \cdot r_1 = s_k \cdot \ldots \cdot s_1$ and $\ftbF r_n \cdot \ldots \cdot \ftbF r_1 \leq \ftbF s_k \cdot \ldots \cdot \ftbF s_1$.
\end{definition}

In Section~\ref{sec:f-actions} we have seen that every \(\SET\)-functor that
admits a normal lax extension preserves 1/4-iso pullbacks, or equivalently, it is monotone on difunctional relations (\autoref{p:64}).
As we show next, the latter condition is also equivalent to satisfying the criterion of \autoref{cor:least_flat} for pairs of composable relations.

\begin{proposition}
	\label{p:300}
	Let $\ftF \colon \SET \to \SET$ be a functor.
	The following clauses are equivalent:
	\begin{enumerate}[(i)]
		\item \label{p:601} The functor $\ftF \colon \SET \to \SET$ preserves 1/4-iso pullbacks.
		\item \label{p:602} For all relations $r_1 \colon X \relto Y$, $r_2 \colon Y \relto X$ such that $r_2 \cdot r_1 \leq  1_X$, $\ftbF r_2 \cdot \ftbF r_1 \leq 1_{\ftF X}$.
		\item \label{p:603} For all relations $r_1 \colon X \relto Y$, $r_2 \colon Y \relto X$ such that $r_2 \cdot r_1 =  1_X$, $\ftbF r_2 \cdot \ftbF r_1 \leq 1_{\ftF X}$.
	\end{enumerate}
\end{proposition}

Now, suppose that we want to extend the previous result in inductive style to composable triples of relations.
Due to the next lemma, a simple idea  to reduce the case of composable triples to the case of composable pairs of relations is to take the difunctional closure of the second relation in the sequence.

\begin{lemma}
	\label{p:610}
	Let $r_1 \colon X_0 \relto X_1$ , $r_2 \colon X_1 \relto X_2$ and $r_3 \colon X_2 \relto X_3$ be relations given by spans that form the base of the commutative diagram
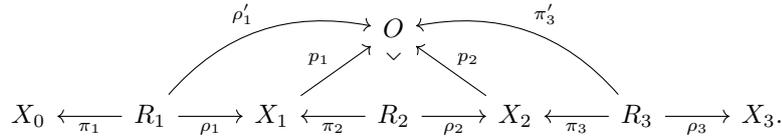
\begin{center}
	\begin{tikzcd}
	    &      &     & O   &     &     &   \\
		X_0 & R_1  & X_1 & R_2 & X_2 & R_3 & X_3. \\
		\ar[from=2-2, to=2-1, "\pi_1"]
		\ar[from=2-2, to=2-3, "\rho_1"']
		\ar[from=2-4, to=2-3, "\pi_2"]
		\ar[from=2-4, to=2-5, "\rho_2"']
		\ar[from=2-6, to=2-5, "\pi_3"]
		\ar[from=2-6, to=2-7, "\rho_3"']
		\ar[from=2-3, to=1-4, "p_1"]
		\ar[from=2-5, to=1-4, "p_2"']
		\ar[from=2-2, to=1-4, bend left, "\rho_1'"]
		\ar[from=2-6, to=1-4, bend right, "\pi_3'"']
		\ar[from=2-4, to=1-4, phantom, very near end, "\rotatebox{45}{$\llcorner$}"]
	\end{tikzcd}
\end{center}
Then, with $r_1' \colon X \relto O$ and  $r_3' \colon O \relto X_3$ defined by the spans $X_0 \xfrom{\pi_1} R_1 \xto{\rho_1'} O$ and $X_0 \xfrom{\pi_3'} R_3 \xto{\rho_3} X_3$, respectively, $\ftbF r_3 \cdot \ftbF r_2 \cdot \ftbF r_1 \leq \ftbF r_3 \cdot \ftbF \hat{r}_2 \cdot \ftbF r_1 \leq \ftbF r_3' \cdot \ftbF r_1'$.
\end{lemma}
Indeed, let $r_1 \colon X \relto X_1$ , $r_2 \colon X_1 \relto X_2$ and $r_3 \colon X_2 \relto X$ be relations such that $r_3 \cdot r_2 \cdot r_1 = 1_X$.
Then, by \autoref{p:300} and \autoref{p:610}, we conclude that $\ftbF r_3 \cdot \ftbF r_2 \cdot \ftbF r_1 \leq 1_{\ftF X}$ once we show that $r_3' \cdot r_1' = 1_X$.
Of course, in general, this does not hold.
Consider the following example where the arrows depict pairs of related elements.

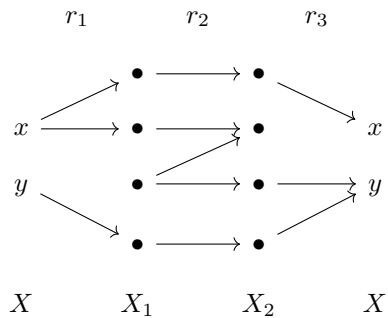
\begin{center}
	\begin{tikzcd}[row sep = small]
		\ar[phantom]{r}{r_{1}}
		& \ \ar[phantom]{r}{r_{2}}
		& \ \ar[phantom]{r}{r_{3}}
		& \ \\
		\
		& \bullet \ar{r}
		& \bullet \ar{dr}
		& \ \\
		x \ar{r}\ar{ur}
		& \bullet\ar{r}
		& \bullet
		& x
		\\
		y \ar{dr}
		& \bullet\ar{r}\ar{ur}
		& \bullet\ar{r}
		& y
		\\
		\
		& \bullet \ar{r}
		& \bullet \ar{ur}
		& \ \\
		X
		& X_1
		& X_2
		& X
	\end{tikzcd}
\end{center}
By taking the difunctional closure $\hat{r}_2$ of $r_2$ we get
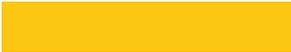
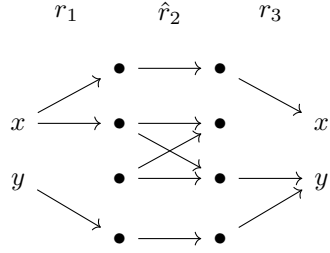
\begin{center}
	\begin{tikzcd}[row sep = small]
		\ar[phantom]{r}{r_{1}}
		& \ \ar[phantom]{r}{\hat{r}_{2}}
		& \ \ar[phantom]{r}{r_{3}}
		& \ \\
		\
		& \bullet \ar{r}
		& \bullet \ar{dr}
		& \ \\
		x\ar{r} \ar{ur}
		& \bullet\ar{r}\ar{dr}
		& \bullet
		& x
		\\
		y \ar{dr}
		& \bullet\ar{r}\ar{ur}
		& \bullet\ar{r}
		& y
		\\
		\
		& \bullet \ar{r}
		& \bullet \ar{ur}
		& \ \\
	\end{tikzcd}
\end{center}
So, $r_3 \cdot \hat{r}_2 \cdot r_1 = r_3' \cdot r_1'$ is not a subidentity.
Now the property of preserving 1/4-iso pullbacks is helpful again.
As we will see in \autoref{p:19}, under this condition, the sequence below is a Barr upper bound of the first one and it is obtained from it by ``splitting'' where necessary the elements of \(X_1\) that do not belong to the codomain of \(r_1\) and the elements of \(X_2\) that do not belong to the domain of \(r_3\).

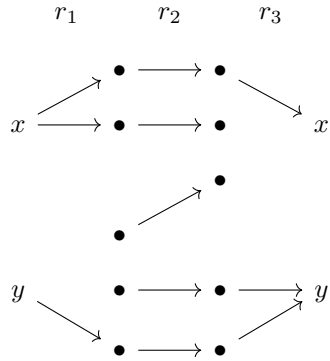
\begin{center}
	\begin{tikzcd}[row sep=small]
		\ar[phantom]{r}{r_{1}}
		& \ \ar[phantom]{r}{r_{2}}
		& \ \ar[phantom]{r}{r_{3}}
		& \ \\
		\
		& \bullet \ar{r}
		& \bullet \ar{dr}
		& \ \\
		x\ar{r} \ar{ur}
		& \bullet\ar{r}
		& \bullet
		& x
		\\
		&
		& \bullet
		&
		\\
		& \bullet\ar{ur}
		&
		&
		\\
		y \ar{dr}
		& \bullet\ar{r}
		& \bullet\ar{r}
		& y
		\\
		\
		& \bullet \ar{r}
		& \bullet \ar{ur}
		& \ \\
	\end{tikzcd}
\end{center}
In this situation we can apply the difunctional closure to $r_2$ (which in this particular example is already difunctional) to reduce the number of relations as discussed in \autoref{p:610}.

\begin{lemma}
	\label{p:19}
	Let $\ftF \colon \SET \to \SET$ be a functor that preserves 1/4-iso pullbacks, and let $r_1 \colon X \relto Y$, $r_2 \colon Y \relto Z$ and $r_3 \colon Z \relto W$ be relations.
   	Then, there are relations $s_1 \colon X \relto Y'$, $s_2 \colon Y' \relto Z'$ and $s_3 \colon Z' \relto W$ such that $s_1, s_2, s_3$ is a Barr upper bound of $r_1, r_2, r_3$ and
	\begin{enumerate}[wide]
		\item for all \(y,y' \in Y'\) and all \(z \in Z'\), if \(y \neq y'\), \(y \mathrel{s_2} z\) and \(y'\mathrel{s_2} z\),
		      then \(z \in \dom(s_3)\);
		\item for all \(y \in Y'\) and \(z,z' \in Z'\), if \(z \neq z'\), \(y \mathrel{s_2} z\) and \(y \mathrel{s_2}
		      z'\), then \(y \in \cod(s_1)\).
	\end{enumerate}
\end{lemma}

The  previous lemma essentially closes the argument that we have been crafting so far.

\begin{theorem}
	\label{p:301}
	Let $\ftF \colon \SET \to \SET$ be a functor.
	The following clauses are equivalent:
	\begin{enumerate}[(i)]
		\item \label{p:701} The functor $\ftF \colon \SET \to \SET$ preserves 1/4-iso pullbacks.
		\item \label{p:702} For all relations $r_1 \colon X \relto Y$, $r_2 \colon Y \relto Z$ and $r_3 \colon Z \relto X$ such that $r_3 \cdot r_2 \cdot r_1 \leq  1_X$, $\ftbF r_3 \cdot \ftbF r_2 \cdot \ftbF r_1 \leq 1_{\ftF X}$.
		\item \label{p:703} For all relations $r_1 \colon X \relto Y$, $r_2 \colon Y \relto Z$ and $r_3 \colon Z \relto X$ such that $r_3 \cdot r_2 \cdot r_1 =  1_X$, $\ftbF r_3 \cdot \ftbF r_2 \cdot \ftbF r_1 \leq 1_{\ftF X}$.
	\end{enumerate}
\end{theorem}

However, as we see next, \autoref{p:301} is as far as we can go under the assumption of 1/4-iso pullbacks preservation.
In other words, the fact that a $\SET$-functor preserves 1/4-iso pullbacks is \emph{not} sufficient to conclude that it admits a normal lax extension.

\begin{example}
	\label{p:920}
	Let us define a functor $\ftF \colon \SET \to \SET$ as a quotient of\/ $\coprod_{n\in \{f,g\}} \{n\}\times X^5\cong X^5 + X^5$ under the equivalence defined by the clauses:
	\begin{flalign*}
		&&f(y,x,z,x,t)\sim\; & f(y',x,z',x,t') & f(t,x,x,y,y)\sim\; & f(t',x,x,y,y)&\\
		&&g(y,x,z,x,t)\sim\; & g(y',x,z',x,t') & g(x,x,y,y,t)\sim\; & g(x,x,y,y,t')&\\
		&&f(y,x,z,x,t)\sim\; & g(y',x,z',x,t') & f(t,x,z,y,z)\sim\; & g(t,x,t,y,z)&
	\end{flalign*}
where $f(x_1,\ldots,x_5)$ and $g(x_1,\ldots,x_5)$ denote the corresponding elements
$(f,x_1,\ldots,x_5)\comma (g,x_1,\ldots,x_5)\in\coprod_{n\in \{f,g\}} \{n\}\times X^5$.
Let $2 = \{x,y\}$ and consider the composable sequence of relations depicted below.
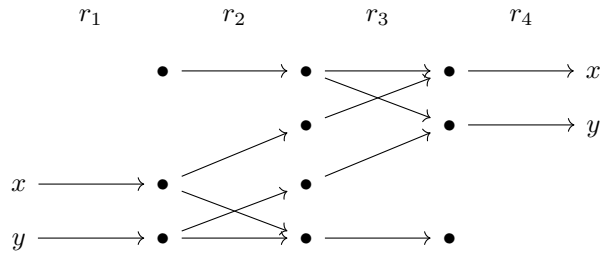
\begin{center}
		\begin{tikzcd}[column sep=4em, row sep=2ex]
			\	& \				&	\				&	\			  & \ \\
			  & \bullet & \bullet & \bullet & x \\
			  &         & \bullet & \bullet & y \\
			x & \bullet & \bullet &             \\
			y & \bullet & \bullet & \bullet
		  \ar[from=1-1, to=1-2, phantom, "r_1"]
  		\ar[from=1-2, to=1-3, phantom, "r_2"]
  		\ar[from=1-3, to=1-4, phantom, "r_3"]
  		\ar[from=1-4, to=1-5, phantom, "r_4"]
			\ar[from=4-1, to=4-2]
			\ar[from=5-1, to=5-2]
			\ar[from=2-2, to=2-3]
			\ar[from=4-2, to=3-3]
			\ar[from=4-2, to=5-3]
			\ar[from=5-2, to=5-3]
			\ar[from=5-2, to=4-3]
			\ar[from=2-3, to=2-4]
			\ar[from=2-3, to=3-4]
			\ar[from=3-3, to=2-4]
			\ar[from=4-3, to=3-4]
			\ar[from=5-3, to=5-4]
			\ar[from=2-4, to=2-5]
			\ar[from=3-4, to=3-5]
		\end{tikzcd}
\end{center}
Then, $\ftF$ preserves 1/4-iso pullbacks and $r_4 \cdot r_3 \cdot r_2 \cdot r_1 = 1_2$, however, $\ftbF r_4 \cdot \ftbF r_3 \cdot  \ftbF r_2 \cdot \ftbF r_1 \not\leq 1_{\ftF 2}$.
\end{example}

\subsection{The case of functors that weakly preserve 4/4-epi pullbacks}
\label{ssec:4/4-epi}

From \autoref{p:301} it basically follows that a functor that weakly preserves 1/4-iso pullbacks and 4/4-epi  pullbacks admits a normal lax extension.
But to see this, first we need to sharpen \autoref{cor:least_flat}.
The goal is to show that it suffices to consider composable sequences of relations where  all relations other than the first and the last are total and surjective.
To illustrate how we achieve this, let us consider the sequence of relations depicted below.

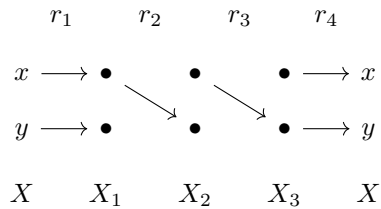
\begin{center}
	\begin{tikzcd}[row sep=small, column sep=small]
		\ar[phantom]{r}{r_{1}}
		& \ \ar[phantom]{r}{r_{2}}
		& \ \ar[phantom]{r}{r_{3}}
		& \ \ar[phantom]{r}{r_{4}}
		& \
		\\
		& \bullet \ar{r}
		& \bullet \ar{r}
		& \bullet \ar{dr}
		&
		\\
		x\ar{r} \ar{ur}
		&\bullet \ar{dr}
		& \bullet\ar{dr}
		& \bullet\ar{r}
		& x
		\\
		y \ar{r} \ar{dr}
		& \bullet
		& \bullet
		& \bullet\ar{r}
		& y
		\\
		& \bullet \ar{r}
		& \bullet \ar{r}
		& \bullet \ar{ur}
		&
		\\
		X
		& X_1
		& X_2
		& X_3
		& X
	\end{tikzcd}
\end{center}
Then, by adding new elements $0$ and $1$ to $X_1,X_2$ and $X_3$ we can extend this sequence to the sequence

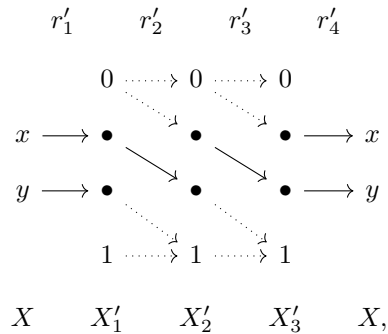
\begin{center}
	\begin{tikzcd}[row sep=small, column sep=small]
		\ar[phantom]{r}{r'_{1}}
		& \ \ar[phantom]{r}{r'_{2}}
		& \ \ar[phantom]{r}{r'_{3}}
		& \ \ar[phantom]{r}{r'_{4}}
		& \
		\\
		& \bullet \ar{r}
		& \bullet \ar{r}
		& \bullet \ar[ddr, bend left]
		&
		\\
		& 0 \ar[dotted]{r} \ar[dotted]{dr}
		& 0 \ar[dotted]{r} \ar[dotted]{dr}
		& 0
		&
		\\
		x\ar{r} \ar[uur, bend left]
		&\bullet \ar{dr}
		& \bullet\ar{dr}
		& \bullet\ar{r}
		& x
		\\
		y \ar{r} \ar[ddr, bend right]
		& \bullet\ar[dotted]{dr}
		& \bullet\ar[dotted]{dr}
		& \bullet\ar{r}
		& y
		\\
		& 1 \ar[dotted]{r}
		& 1\ar[dotted]{r}
		& 1
		&
		\\
		& \bullet \ar{r}
		& \bullet \ar{r}
		& \bullet \ar[uur, bend right]
		&
		\\
		X
		& X_1'
		& X_2'
		& X_3'
		& X,
	\end{tikzcd}
\end{center}
where the dotted arrows indicate pairs of elements that were added to the corresponding relation as follows:
for $i=2,3$, $r'_i$ relates $0 \in X'_{i-1}$ to every element of $X_i \cup \{0\}$ that does not belong to the codomain of $r_i$  and relates every element of $X_{i-1} \cup \{1\}$ that does not belong to the domain of $r_i$ to $1 \in X'_i$.
In this way, we guarantee that $r'_2$ and $r_3'$ are total and surjective and that $r_4' \cdot r_3' \cdot r_2' \cdot r_1' = r_4 \cdot r_3 \cdot r_2 \cdot r_1 = 1_X$.
We could have extended $r_2$ and $r_3$ to total and surjective relations by adding just a single element $\ast$ to $X_1$, $X_2$ and $X_3$ that would simultaneously take the role of $0$ and $1$.
However, composing the resulting sequence of relations would not yield the identity relation:
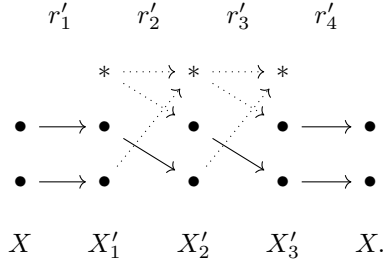
\begin{center}
	\begin{tikzcd}[row sep=small, column sep=small]
		\ar[phantom]{r}{r'_{1}}
		& \ \ar[phantom]{r}{r'_{2}}
		& \ \ar[phantom]{r}{r'_{3}}
		& \ \ar[phantom]{r}{r'_{4}}
		& \
		\\
		& \bullet \ar{r}
		& \bullet \ar{r}
		& \bullet \ar[ddr, bend left]
		&
		\\
		&  \ast \ar[dotted]{r} \ar[dotted]{dr}
		&  \ast \ar[dotted]{r} \ar[dotted]{dr}
		&  \ast
		&
		\\
		x \ar{r} \ar[uur, bend left]
		&\bullet \ar{dr}
		& \bullet\ar{dr}
		& \bullet\ar{r}
		& x
		\\
		y \ar{r} \ar[dr]
		& \bullet\ar[dotted]{uur}
		& \bullet\ar[dotted]{uur}
		& \bullet\ar{r}
		& y
		\\
		& \bullet \ar{r}
		& \bullet \ar{r}
		& \bullet \ar[ur]
		&
		\\
		X
		& X_1'
		& X_2'
		& X_3'
		& X.
	\end{tikzcd}
\end{center}
In other words, by splitting $\ast$ in two elements $0$ and $1$, the former to make the relations $r_2$ and $r_3$ surjective and the latter to make them total, we obtain a subidentity because we never create paths between elements of $X_1$ that are not part of the domain of $r_2$ and elements of $X_3$ that are not part of the codomain of $r_3$.
In the next lemma we formalize this procedure for arbitrary composable sequences of relations and show that it yields Barr upper bounds.

\begin{lemma}
	\label{p:600}
	A functor \(\ftF \colon \SET \to \SET\) that preserves 1/4-iso pullbacks admits a normal lax extension iff for every composable sequence of relations \(r_1,\ldots,r_n\) such that $n\geq 4$ and $r_2, \ldots, r_{n-1}$ are total and surjective, whenever \(r_n \cdot \ldots \cdot r_1 = 1_X\), for some set $X$, then $\ftbF r_n \cdot \ldots \cdot \ftbF r_1 \leq 1_{\ftF X}$.
\end{lemma}

\begin{remark}
	In a composable sequence of relations that satisfies the conditions of \autoref{p:600} the first relation is necessarily total while the last one is necessarily surjective.
\end{remark}

Now, our first main result follows straightforwardly.
Since the composite of total and surjective relations is total and surjective, due to the following fact, every composable sequence of relations where all relations other than the first and the last are total and surjective admits a Barr upper bound consisting of three relations.

\begin{proposition}
	\label{p:960}
	A functor $\ftF \colon \SET \to \SET$ weakly preserves 4/4-epi pullbacks iff for all relations $r \colon X \relto Y$ and $s \colon Y \relto Z$, whenever $r$ is surjective and $s$ is total, $\ftbF s \cdot \ftbF r = \ftbF(s \cdot r)$.
\end{proposition}

\begin{theorem}
	\label{p:3}
	A \(\SET\)-functor that weakly preserves 1/4-iso pullbacks and 4/4-epi weak pullbacks admits a normal lax extension.
\end{theorem}

\begin{remark}
  Preservation of 4/4-epi pullbacks plays a role in the analysis of
  interpolation in coalgebraic logic~\cite{SeifanEA17}. In particular,
  this analysis implies that given a separating set~$\Lambda$ of
  monotone predicate liftings for a finite-set-preserving functor~$~\ftF$,
  which induces an expressive modal logic~$\mathcal{L}(\Lambda)$ for
  $~\ftF$-coalgebras, the logic $\mathcal{L}(\Lambda)$ has interpolation
  iff~$\ftF$ weakly preserves 4/4-epi
  pullbacks~\cite[Theorem~37]{SeifanEA17}. In connection with the fact
  that a functor has a normal lax extension iff it has a separating
  set of monotone predicate liftings~\cite{MV15}, we obtain the
  following application of \autoref{p:3} and \autoref{p:450}: A
  finite-set preserving functor~$\ftF$ has a separating set of monotone
  predicate liftings such that the associated modal logic has uniform
  interpolation iff~$\ftF$ weakly preserves 1/4-iso pullbacks and 4/4-epi
  pullbacks.
\end{remark}

\subsection{The case of functors that preserve 1/4-mono pullbacks}
\label{ssec:1/4-mono}

To obtain \autoref{p:3}, we refined \autoref{cor:least_flat} to composable sequences of relations where all relations \emph{other than} the first and the last are total and surjective.
And to achieve this in \autoref{p:600},  given a composable sequence of relations, we \emph{added} pairs of related elements to the relations in the sequence.
In the sequel, we will show that every functor that preserves 1/4-mono pullbacks admits a normal lax extension.
We will see that for these functors it is even possible to refine \autoref{cor:least_flat} to composable sequences of relations where \emph{all} relations are total and surjective. However, we will achieve this in \autoref{p:401} below by, given a composable sequence of relations, \emph{removing} pairs of related elements from the relations in the sequence.
Our proof strategy is justified by the next fact.

\begin{proposition}
	\label{p:100}
	A functor $\ftF \colon \SET \to \SET$ preserves 1/4-mono pullbacks iff for all relations $r \colon X \relto Y$ and $s \colon Y \relto Z$, whenever $r$ is the converse of a partial function or $s$ is a partial function, $\ftbF s \cdot \ftbF r = \ftbF(s \cdot r)$.
\end{proposition}

This result enables a ``look ahead and behind'' strategy for \autoref{cor:least_flat}.
The idea is that, given a composable sequence of relations $r_1, \ldots, r_n$ such that $r_n \cdot \ldots \cdot r_1 = 1_X$, then, with $r_i \colon X_{i-1} \relto X_i$ being a relation in the sequence, removing the elements of $X_i$ that do not belong to the codomain of $r_i \cdot \ldots \cdot r_1$ \emph{or} do not belong to the domain of $r_n \cdot \ldots \cdot r_{i+1}$ yields a Barr upper bound of our original sequence.
For instance, consider the composable sequence of relations depicted in \autoref{p:920}, which we used to show that there are functors that preserve 1/4-iso pullbacks but do not admit a normal lax extension.
In the next lemma, in particular, we show that for functors that preserve 1/4-mono pullbacks the sequence below of total and surjective relations is a Barr upper bound of this one.
The dotted arrows represent pairs of related elements that were removed, and the grey boxes represent the elements of each set that are \emph{not} removed.

\begin{center}
		\begin{tikzcd}[column sep=4em, row sep=2ex, /tikz/execute at end picture={
						\begin{pgfonlayer}{background}
							\node (large) [rectangle, draw=none, fill=lightgray, rounded corners, fit=(\tikzcdmatrixname-4-1) (\tikzcdmatrixname-5-1)] {};
							\node (large) [rectangle, draw=none, fill=lightgray, rounded corners, fit=(\tikzcdmatrixname-4-2) (\tikzcdmatrixname-5-2)] {};
							\node (large) [rectangle, draw=none, fill=lightgray, rounded corners, fit=(\tikzcdmatrixname-3-3) (\tikzcdmatrixname-4-3)] {};
							\node (large) [rectangle, draw=none, fill=lightgray, rounded corners, fit=(\tikzcdmatrixname-2-4) (\tikzcdmatrixname-3-4)] {};
							\node (large) [rectangle, draw=none, fill=lightgray, rounded corners, fit=(\tikzcdmatrixname-2-5) (\tikzcdmatrixname-3-5)] {};
						\end{pgfonlayer}%
					}]
			\	& \				&	\				&	\			  & \ \\
			  & \bullet & \bullet & \bullet & x \\
			  &         & \bullet & \bullet & y \\
			x & \bullet & \bullet &             \\
			y & \bullet & \bullet & \bullet
		  \ar[from=1-1, to=1-2, phantom, "r'_1"]
  		\ar[from=1-2, to=1-3, phantom, "r'_2"]
  		\ar[from=1-3, to=1-4, phantom, "r'_3"]
  		\ar[from=1-4, to=1-5, phantom, "r'_4"]
			\ar[from=4-1, to=4-2]
			\ar[from=5-1, to=5-2]
			\ar[from=2-2, to=2-3, dotted]
			\ar[from=4-2, to=3-3]
			\ar[from=4-2, to=5-3, dotted]
			\ar[from=5-2, to=5-3, dotted]
			\ar[from=5-2, to=4-3]
			\ar[from=2-3, to=2-4, dotted]
			\ar[from=2-3, to=3-4, dotted]
			\ar[from=3-3, to=2-4]
			\ar[from=4-3, to=3-4]
			\ar[from=5-3, to=5-4, dotted]
			\ar[from=2-4, to=2-5]
			\ar[from=3-4, to=3-5]
		\end{tikzcd}
\end{center}

\begin{lemma}
	\label{p:401}
	A functor $\ftF \colon \SET \to \SET$ that preserves 1/4-mono pullbacks admits a normal lax extension if for every composable sequence of total and surjective relations $r_1,\ldots,r_n$, whenever $r_n \cdot \ldots \cdot r_1 = 1_X$ for some set $X$, then $\ftbF r_n \cdot \ldots \cdot \ftbF r_1 \leq 1_{\ftF X}$.
\end{lemma}

It turns out that the sufficient condition of the previous lemma is actually satisfied by every $\SET$-functor that preserves 1/4-iso pullbacks.
Indeed, due to the next result, \autoref{p:610} and the fact that surjections are stable under pushouts, every composable sequence of total and surjective relations whose composite is an identity admits a Barr upper bound consisting of three relations.

\begin{lemma}
	\label{p:900}
	Let $r_1 \colon X \relto X_1$, $r_2 \colon X_1 \relto X_2$ and $r_3 \colon X_2 \relto X$ be a composable sequence of total and surjective relations, and let $\hat{r}_2 \colon X_1 \relto X_2$ be the difunctional closure of $r_2$.
	If $r_3 \cdot r_2 \cdot r_1 = 1_X$, then $r_3 \cdot \hat{r}_2 \cdot r_1 = 1_X$.
\end{lemma}

\begin{proposition}
	Let $\ftF \colon \SET \to \SET$ be a functor that preserves 1/4-iso pullbacks, and let $r_1,\ldots,r_n$ be a composable sequence of total and surjective relations.
	If $r_n \cdot \ldots \cdot r_1 = 1_X$ for some set $X$, then $\ftbF r_n \cdot \ldots \cdot \ftbF r_1 \leq 1_{\ftF X}$.
\end{proposition}

Therefore,

\begin{theorem}
	\label{thm:main}
	Every \(\SET\)-functor that preserves 1/4-mono pullbacks admits a normal lax extension.
\end{theorem}

In particular, since in \autoref{expl:functors}(\ref{p:403}) we have seen that for (commutative) monoid-valued functors preserving 1/4-mono pullbacks is equivalent to preserving 1/4-iso pullbacks, as a consequence of \autoref{thm:main} and \autoref{p:450} we obtain:

\begin{corollary}\label{cor:monoid-lax}
	A (commutative) monoid-valued functor admits a normal lax extension iff the monoid is positive.
\end{corollary}

\begin{remark}
  The above result may be equivalently stated as saying that a
  monoid-valued functor has a separating set of monotone predicate
  liftings iff the monoid is positive. In this formulation, it
  improves on a previous result effectively stating the same
  equivalence under the additional assumption that the monoid is
  refinable~\cite[Proposition 22]{SeifanEA17}. 
  For every monoid~$M$, one has a preorder on~$M$ given by $m\ge n$
  iff $\exists k \in M.\,m=n+k$, which is a partial order whenever the monoid is cancellative and positive.
  It is then clear that one has a separating
  set of monotone predicate liftings $\Diamond_m$, for~$m\in M$, defined by
  $\Diamond_m(A)=\{\mu\in M^{(X)}\mid \mu(A)\ge m\}$ where we write
  $\mu(A)=\sum_{x\in A}\mu(x)$. The arising normal lax extension is
  given for $r\colon X\relto Y$, $\mu\in M^{(X)}$, $\nu\in M^{(Y)}$ by
  $\mu\mathrel{Lr}\nu$ iff $\nu(r[A])\ge \mu(A)$ for all
  $A\subseteq X$ and symmetrically, much like for probabilistic
  transition systems (\autoref{expl:lax}(\ref{item:lax-prob})). For
  non-cancellative positive monoids, the description of the normal lax
  extension and the separating set of monotone predicate liftings
  whose existence are guaranteed by \autoref{cor:monoid-lax} is in
  general more involved. In particular, the predicate liftings
  $\Diamond_m$ described above may fail to be separating, as
  witnessed, for instance, by the commutative additive monoid $\{0,1,2\}$ with $1+2=1$. 
  Specifically, $\mu,\nu\in M^{(\{\star\})}$ given by $\mu(\star)=1$ and $\nu(\star)=2$ cannot be
  distinguished.
\end{remark}

The class of $\SET$-functors that admit a normal lax extension is closed under subfunctors and several natural constructions such as the sum of functors.
This makes it easy to extend the reach of our sufficient conditions, but it also shows that it is easy to provide examples of functors that admit a normal lax extension and do not weakly preserve 1/4-mono pullbacks nor 4/4-epi pullbacks.
A quick example is the functor given by the sum of the functor $(-)^3_2$ and the monotone neighbourhood functor.
To conclude this section, we present a less obvious example that is constructed analogously to \autoref{p:920}.
Notice that, as we have seen in \autoref{p:860}(\ref{p:861}), the class of functors that admit a normal lax extension is not closed under quotients.

\begin{example}
For any set $X$, let $\ftF X$
be the
quotient of\/ $X^3$ under the equivalence relation~$\sim$ defined by the
clauses $(x,x,y)\sim (x,x,x)\sim (y,x,x)$. This yields a functor $\ftF \colon \SET \to \SET$
that neither weakly preserves 1/4-mono pullbacks nor 4/4-epi pullbacks, however,
$\ftF$ admits a normal lax extension.
\pnnote{@Sergey: add details to the appendix}
\end{example}

\section{Conclusions}

\noindent Normal lax extensions of functors play a dual role in the
coalgebraic modelling of reactive systems, on the one hand allowing
for good notions of bisimulations on functor coalgebras and on the
other hand guaranteeing the existence of expressive temporal
logics. We have shown on the one hand that every functor admitting a
lax extension preserves 1/4-iso pullbacks, and on the other hand that
a functor admits a normal lax extension if it weakly preserves either
1/4-iso pullbacks and 4/4-epi pullbacks or inverse images. These
results improve on previous
results~\cite{KurzLeal12,MartiVenema12,MV15}, which combine to imply
that weak-pullback-preserving functors admit normal lax
extensions. One application of our results implies, roughly, that a
given type of monoid-weighted transition systems admits a good notion
of bisimulation iff the monoid is positive.

The most obvious issue for future work is to close the remaining gap,
i.e.\ to give a necessary and sufficient criterion for the existence
of normal lax extensions in terms of limit preservation. Additionally,
the structure of the lattice of normal lax extensions of a functor merits
attention, in the sense that larger lax extensions induce more
permissive notions of bisimulation.

\bibliography{flatbibl}

\providecommand{\noopsort}[1]{}
\begin{thebibliography}{10}

\bibitem{AM89}
Peter Aczel and Nax~Paul Mendler.
\newblock A final coalgebra theorem.
\newblock In David~H. Pitt, David~E. Rydeheard, Peter Dybjer, Andrew~M. Pitts,
  and Axel Poign{\'{e}}, editors, {\em Category Theory and Computer Science,
  Manchester, UK, September 5-8, 1989, Proceedings}, volume 389 of {\em Lecture
  Notes in Computer Science}, pages 357--365. Springer, 1989.
\newblock \href {https://doi.org/10.1007/BFB0018361}
  {\path{doi:10.1007/BFB0018361}}.

\bibitem{AHS90}
Ji{\v{r}\'i} Ad{\'a}mek, Horst Herrlich, and George~E. Strecker.
\newblock {\em Abstract and concrete categories: {T}he joy of cats}.
\newblock John Wiley \& Sons Inc., 1990.
\newblock Republished in: Reprints in Theory and Applications of Categories,
  No. 17 (2006) pp.~1--507.
\newblock URL: \url{http://tac.mta.ca/tac/reprints/articles/17/tr17abs.html}.

\bibitem{BackhouseEA91}
Roland~Carl Backhouse, Peter~J. de~Bruin, Paul~F. Hoogendijk, Grant Malcolm,
  Ed~Voermans, and Jaap van~der Woude.
\newblock Polynomial relators (extended abstract).
\newblock In Maurice Nivat, Charles Rattray, Teodor Rus, and Giuseppe Scollo,
  editors, {\em Algebraic Methodology and Software Technology, AMAST 1991},
  Workshops in Computing, pages 303--326. Springer, 1991.

\bibitem{BaldanEA18}
Paolo Baldan, Filippo Bonchi, Henning Kerstan, and Barbara K{\"{o}}nig.
\newblock Coalgebraic behavioral metrics.
\newblock {\em Log.\ Methods Comput.\ Sci.}, 14(3), 2018.
\newblock \href {https://doi.org/10.23638/LMCS-14(3:20)2018}
  {\path{doi:10.23638/LMCS-14(3:20)2018}}.

\bibitem{Barr70}
Michael Barr.
\newblock Relational algebras.
\newblock In {\em Reports of the Midwest Category Seminar IV}, number 137 in
  Lect.\ Notes Math., pages 39--55. Springer, 1970.
\newblock \href {https://doi.org/10.1007/BFb0060439}
  {\path{doi:10.1007/BFb0060439}}.

\bibitem{Bar70}
Michael Barr.
\newblock Relational algebras.
\newblock In {\em Reports of the Midwest Category Seminar IV}, pages 39--55.
  Springer, 1970.
\newblock \href {https://doi.org/10.1007/bfb0060439}
  {\path{doi:10.1007/bfb0060439}}.

\bibitem{Barr93}
Michael Barr.
\newblock Terminal coalgebras in well-founded set theory.
\newblock {\em Theor.\ Comput. Sci.}, 114(2):299--315, 1993.
\newblock \href {https://doi.org/10.1016/0304-3975(93)90076-6}
  {\path{doi:10.1016/0304-3975(93)90076-6}}.

\bibitem{BirdDeMoor97}
Richard~S. Bird and Oege de~Moor.
\newblock {\em Algebra of programming}.
\newblock Prentice Hall International series in computer science. Prentice
  Hall, 1997.

\bibitem{BonchiEA12}
Filippo Bonchi, Marcello~M. Bonsangue, Michele Boreale, Jan J. M.~M. Rutten,
  and Alexandra Silva.
\newblock A coalgebraic perspective on linear weighted automata.
\newblock {\em Inf.\ Comput.}, 211:77--105, 2012.
\newblock \href {https://doi.org/10.1016/J.IC.2011.12.002}
  {\path{doi:10.1016/J.IC.2011.12.002}}.

\bibitem{Chellas80}
Brian~F. Chellas.
\newblock {\em Modal Logic - An Introduction}.
\newblock Cambridge University Press, 1980.
\newblock \href {https://doi.org/10.1017/CBO9780511621192}
  {\path{doi:10.1017/CBO9780511621192}}.

\bibitem{CirsteaEA11}
Corina C{\^{\i}}rstea, Clemens Kupke, and Dirk Pattinson.
\newblock {EXPTIME} tableaux for the coalgebraic mu-calculus.
\newblock {\em Log.\ Methods Comput.\ Sci.}, 7(3), 2011.
\newblock \href {https://doi.org/10.2168/LMCS-7(3:3)2011}
  {\path{doi:10.2168/LMCS-7(3:3)2011}}.

\bibitem{ClementinoEA14}
Maria~Manuel Clementino, Dirk Hofmann, and George Janelidze.
\newblock The monads of classical algebra are seldom weakly cartesian.
\newblock {\em J.\ Homotopy and Related Structures}, 9(1):175--197, 2014.

\bibitem{Gavazzo18}
Francesco Gavazzo.
\newblock Quantitative behavioural reasoning for higher-order effectful
  programs: Applicative distances.
\newblock In Anuj Dawar and Erich Gr{\"{a}}del, editors, {\em Logic in Computer
  Science, {LICS} 2018}, pages 452--461. {ACM}, 2018.
\newblock \href {https://doi.org/10.1145/3209108.3209149}
  {\path{doi:10.1145/3209108.3209149}}.

\bibitem{GoncharovEA23}
Sergey Goncharov, Dirk Hofmann, Pedro Nora, Lutz Schröder, and Paul Wild.
\newblock A point-free perspective on lax extensions and predicate liftings.
\newblock {\em Mathematical Structures in Computer Science}, page 1–30, 2023.
\newblock \href {https://doi.org/10.1017/S096012952300035X}
  {\path{doi:10.1017/S096012952300035X}}.

\bibitem{GorinSchroder13}
Daniel Gor{\'{\i}}n and Lutz Schr{\"{o}}der.
\newblock Simulations and bisimulations for coalgebraic modal logics.
\newblock In Reiko Heckel and Stefan Milius, editors, {\em Algebra and
  Coalgebra in Computer Science, {CALCO} 2013}, volume 8089 of {\em LNCS},
  pages 253--266. Springer, 2013.
\newblock \href {https://doi.org/10.1007/978-3-642-40206-7\_19}
  {\path{doi:10.1007/978-3-642-40206-7\_19}}.

\bibitem{Gumm20}
H.~Peter Gumm.
\newblock Free-algebra functors from a coalgebraic perspective.
\newblock In Daniela Petrisan and Jurriaan Rot, editors, {\em Coalgebraic
  Methods in Computer Science, {CMCS} 2020}, volume 12094 of {\em LNCS}, pages
  55--67. Springer, 2020.
\newblock \href {https://doi.org/10.1007/978-3-030-57201-3\_4}
  {\path{doi:10.1007/978-3-030-57201-3\_4}}.

\bibitem{GS00}
H.~Peter Gumm and Tobias Schr{\"{o}}der.
\newblock Coalgebraic structure from weak limit preserving functors.
\newblock In Horst Reichel, editor, {\em Coalgebraic Methods in Computer
  Science, {CMCS} 2000, Berlin, Germany, March 25-26, 2000}, volume~33 of {\em
  Electronic Notes in Theoretical Computer Science}, pages 111--131. Elsevier,
  2000.
\newblock \href {https://doi.org/10.1016/S1571-0661(05)80346-9}
  {\path{doi:10.1016/S1571-0661(05)80346-9}}.

\bibitem{GummSchroder01}
H.~Peter Gumm and Tobias Schr{\"{o}}der.
\newblock Monoid-labeled transition systems.
\newblock In Andrea Corradini, Marina Lenisa, and Ugo Montanari, editors, {\em
  Coalgebraic Methods in Computer Science, {CMCS} 2001}, volume 44(1) of {\em
  ENTCS}, pages 185--204. Elsevier, 2001.
\newblock \href {https://doi.org/10.1016/S1571-0661(04)80908-3}
  {\path{doi:10.1016/S1571-0661(04)80908-3}}.

\bibitem{GS05}
H.~Peter Gumm and Tobias Schr{\"{o}}der.
\newblock Types and coalgebraic structure.
\newblock {\em Algebra universalis}, 53(2–3):229--252, August 2005.
\newblock \href {https://doi.org/10.1007/s00012-005-1888-2}
  {\path{doi:10.1007/s00012-005-1888-2}}.

\bibitem{GZ14}
H.~Peter Gumm and Mehdi Zarrad.
\newblock Coalgebraic simulations and congruences.
\newblock In Marcello~M. Bonsangue, editor, {\em Coalgebraic Methods in
  Computer Science - 12th {IFIP} {WG} 1.3 International Workshop, {CMCS} 2014,
  Colocated with {ETAPS} 2014, Grenoble, France, April 5-6, 2014, Revised
  Selected Papers}, volume 8446 of {\em Lecture Notes in Computer Science},
  pages 118--134. Springer, 2014.
\newblock \href {https://doi.org/10.1007/978-3-662-44124-4\_7}
  {\path{doi:10.1007/978-3-662-44124-4\_7}}.

\bibitem{HansenKupke04}
Helle~Hvid Hansen and Clemens Kupke.
\newblock A coalgebraic perspective on monotone modal logic.
\newblock In Jir{\'{\i}} Ad{\'{a}}mek and Stefan Milius, editors, {\em
  Coalgebraic Methods in Computer Science, {CMCS} 2004}, volume 106 of {\em
  ENTCS}, pages 121--143. Elsevier, 2004.
\newblock \href {https://doi.org/10.1016/j.entcs.2004.02.028}
  {\path{doi:10.1016/j.entcs.2004.02.028}}.

\bibitem{HT00}
Wim~H. Hesselink and Albert Thijs.
\newblock Fixpoint semantics and simulation.
\newblock {\em Theor. Comput. Sci.}, 238(1-2):275--311, 2000.
\newblock \href {https://doi.org/10.1016/S0304-3975(98)00176-5}
  {\path{doi:10.1016/S0304-3975(98)00176-5}}.

\bibitem{HST14}
Dirk Hofmann, Gavin~J. Seal, and Walter Tholen, editors.
\newblock {\em Monoidal {T}opology. {A} {C}ategorical {A}pproach to {O}rder,
  {M}etric, and {T}opology}, volume 153 of {\em Encyclopedia of Mathematics and
  its Applications}.
\newblock Cambridge University Press, Cambridge, July 2014.
\newblock Authors: Maria Manuel Clementino, Eva Colebunders, Dirk Hofmann,
  Robert Lowen, Rory Lucyshyn-Wright, Gavin J.\ Seal and Walter Tholen.
\newblock \href {https://doi.org/10.1017/cbo9781107517288}
  {\path{doi:10.1017/cbo9781107517288}}.

\bibitem{HughesJacobs04}
Jesse Hughes and Bart Jacobs.
\newblock Simulations in coalgebra.
\newblock {\em Theor.\ Comput.\ Sci.}, 327(1-2):71--108, 2004.
\newblock \href {https://doi.org/10.1016/J.TCS.2004.07.022}
  {\path{doi:10.1016/J.TCS.2004.07.022}}.

\bibitem{Gumm03}
Thomas Ihringer.
\newblock {\em Algemeine Algebra. Mit einem Anhang \"{u}ber Universelle
  Coalgebra von H.~P.~Gumm}, volume~10 of {\em Berliner Studienreihe zur
  Mathematik}.
\newblock Heldermann Verlag, 2003.

\bibitem{Klin09}
Bartek Klin.
\newblock Structural operational semantics for weighted transition systems.
\newblock In Jens Palsberg, editor, {\em Semantics and Algebraic Specification,
  Essays Dedicated to Peter D. Mosses on the Occasion of His 60th Birthday},
  volume 5700 of {\em LNCS}, pages 121--139. Springer, 2009.
\newblock \href {https://doi.org/10.1007/978-3-642-04164-8\_7}
  {\path{doi:10.1007/978-3-642-04164-8\_7}}.

\bibitem{KupkeEA12}
Clemens Kupke, Alexander Kurz, and Yde Venema.
\newblock Completeness for the coalgebraic cover modality.
\newblock {\em Log. Methods Comput. Sci.}, 8(3), 2012.
\newblock \href {https://doi.org/10.2168/LMCS-8(3:2)2012}
  {\path{doi:10.2168/LMCS-8(3:2)2012}}.

\bibitem{KurzLeal12}
Alexander Kurz and Raul~Andres Leal.
\newblock Modalities in the stone age: {A} comparison of coalgebraic logics.
\newblock {\em Theor.\ Comput.\ Sci.}, 430:88--116, 2012.
\newblock \href {https://doi.org/10.1016/J.TCS.2012.03.027}
  {\path{doi:10.1016/J.TCS.2012.03.027}}.

\bibitem{Lev11}
Paul~Blain Levy.
\newblock Similarity quotients as final coalgebras.
\newblock In Martin Hofmann, editor, {\em Foundations of Software Science and
  Computational Structures, {FOSSACS} 2011}, volume 6604 of {\em LNCS}, pages
  27--41. Springer, 2011.
\newblock \href {https://doi.org/10.1007/978-3-642-19805-2\_3}
  {\path{doi:10.1007/978-3-642-19805-2\_3}}.

\bibitem{MartiVenema12}
Johannes Marti and Yde Venema.
\newblock Lax extensions of coalgebra functors.
\newblock In Dirk Pattinson and Lutz Schr{\"{o}}der, editors, {\em Coalgebraic
  Methods in Computer Science, CMCS 2021}, volume 7399 of {\em LNCS}, pages
  150--169. Springer, 2012.
\newblock \href {https://doi.org/10.1007/978-3-642-32784-1\_9}
  {\path{doi:10.1007/978-3-642-32784-1\_9}}.

\bibitem{MV15}
Johannes Marti and Yde Venema.
\newblock Lax extensions of coalgebra functors and their logic.
\newblock {\em Journal of Computer and System Sciences}, 81(5):880--900, 2015.
\newblock \href {https://doi.org/10.1016/j.jcss.2014.12.006}
  {\path{doi:10.1016/j.jcss.2014.12.006}}.

\bibitem{Milner89}
Robin Milner.
\newblock {\em Communication and concurrency}.
\newblock {PHI} Series in computer science. Prentice Hall, 1989.

\bibitem{Parikh83}
Rohit Parikh.
\newblock Propositional game logic.
\newblock In {\em Foundations of Computer Science, FOCS 1983}, pages 195--200.
  {IEEE} Computer Society, 1983.
\newblock \href {https://doi.org/10.1109/SFCS.1983.47}
  {\path{doi:10.1109/SFCS.1983.47}}.

\bibitem{Pattinson04}
Dirk Pattinson.
\newblock Expressive logics for coalgebras via terminal sequence induction.
\newblock {\em Notre Dame J.\ Formal Log.}, 45(1):19--33, 2004.
\newblock \href {https://doi.org/10.1305/ndjfl/1094155277}
  {\path{doi:10.1305/ndjfl/1094155277}}.

\bibitem{Peleg85}
David Peleg.
\newblock Concurrent dynamic logic (extended abstract).
\newblock In Robert Sedgewick, editor, {\em Symposium on Theory of Computing,
  STOC 1985}, pages 232--239. {ACM}, 1985.
\newblock \href {https://doi.org/10.1145/22145.22172}
  {\path{doi:10.1145/22145.22172}}.

\bibitem{Rig48}
Jacques Riguet.
\newblock Relations binaires, fermetures, correspondances de {Galois}.
\newblock {\em Bulletin de la Société Mathématique de France}, 76:114--155,
  1948.
\newblock \href {https://doi.org/10.24033/bsmf.1401}
  {\path{doi:10.24033/bsmf.1401}}.

\bibitem{Rutten98}
Jan J. M.~M. Rutten.
\newblock Relators and metric bisimulations.
\newblock In Bart Jacobs, Larry Moss, Horst Reichel, and Jan J. M.~M. Rutten,
  editors, {\em Coalgebraic Methods in Computer Science, {CMCS} 1998},
  volume~11 of {\em ENTCS}, pages 252--258. Elsevier, 1998.
\newblock \href {https://doi.org/10.1016/S1571-0661(04)00063-5}
  {\path{doi:10.1016/S1571-0661(04)00063-5}}.

\bibitem{Rutten00}
Jan J. M.~M. Rutten.
\newblock Universal coalgebra: a theory of systems.
\newblock {\em Theor.\ Comput.\ Sci.}, 249(1):3--80, 2000.
\newblock \href {https://doi.org/10.1016/S0304-3975(00)00056-6}
  {\path{doi:10.1016/S0304-3975(00)00056-6}}.

\bibitem{Schroder08}
Lutz Schr{\"{o}}der.
\newblock Expressivity of coalgebraic modal logic: The limits and beyond.
\newblock {\em Theor.\ Comput.\ Sci.}, 390(2-3):230--247, 2008.
\newblock \href {https://doi.org/10.1016/J.TCS.2007.09.023}
  {\path{doi:10.1016/J.TCS.2007.09.023}}.

\bibitem{Sch08}
Lutz Schr{\"{o}}der.
\newblock Expressivity of coalgebraic modal logic: The limits and beyond.
\newblock {\em Theoretical Computer Science}, 390(2-3):230--247, jan 2008.
\newblock \href {https://doi.org/10.1016/j.tcs.2007.09.023}
  {\path{doi:10.1016/j.tcs.2007.09.023}}.

\bibitem{SS08}
Christoph Schubert and Gavin~J. Seal.
\newblock Extensions in the theory of lax algebras.
\newblock {\em Theory and Applications of Categories}, 21(7):118--151, 2008.

\bibitem{Sea05}
Gavin~J. Seal.
\newblock Canonical and op-canonical lax algebras.
\newblock {\em Theory and Applications of Categories}, 14(10):221--243, 2005.
\newblock URL: \url{http://www.tac.mta.ca/tac/volumes/14/10/14-10abs.html}.

\bibitem{SeifanEA17}
Fatemeh Seifan, Lutz Schr{\"{o}}der, and Dirk Pattinson.
\newblock Uniform interpolation in coalgebraic modal logic.
\newblock In Filippo Bonchi and Barbara K{\"{o}}nig, editors, {\em Algebra and
  Coalgebra in Computer Science, {CALCO} 2017}, volume~72 of {\em LIPIcs},
  pages 21:1--21:16. Schloss Dagstuhl -- Leibniz-Zentrum f{\"{u}}r Informatik,
  2017.
\newblock \href {https://doi.org/10.4230/LIPICS.CALCO.2017.21}
  {\path{doi:10.4230/LIPICS.CALCO.2017.21}}.

\bibitem{ThijsThesis}
Albert Thijs.
\newblock {\em Simulation and fixpoint semantics}.
\newblock PhD thesis, University of Groningen, 1996.

\bibitem{Trnkova80}
V{\v{e}}ra Trnkov{\'a}.
\newblock General theory of relational automata.
\newblock {\em Fund.\ Inform.}, 3(2):189--233, 1980.
\newblock \href {https://doi.org/10.3233/FI-1980-3208}
  {\path{doi:10.3233/FI-1980-3208}}.

\bibitem{Trn69}
Věra Trnková.
\newblock Some properties of set functors.
\newblock {\em Commentationes Mathematicae Universitatis Carolinae},
  010(2):323--352, 1969.
\newblock URL: \url{http://eudml.org/doc/16330}.

\bibitem{WildSchroder20}
Paul Wild and Lutz Schr{\"{o}}der.
\newblock Characteristic logics for behavioural metrics via fuzzy lax
  extensions.
\newblock In Igor Konnov and Laura Kov{\'{a}}cs, editors, {\em Concurrency
  Theory, {CONCUR} 2020}, volume 171 of {\em LIPIcs}, pages 27:1--27:23.
  Schloss Dagstuhl -- Leibniz-Zentrum f{\"{u}}r Informatik, 2020.
\newblock \href {https://doi.org/10.4230/LIPICS.CONCUR.2020.27}
  {\path{doi:10.4230/LIPICS.CONCUR.2020.27}}.

\bibitem{WildSchroder22}
Paul Wild and Lutz Schr{\"{o}}der.
\newblock Characteristic logics for behavioural hemimetrics via fuzzy lax
  extensions.
\newblock {\em Log.\ Methods Comput.\ Sci.}, 18(2), 2022.
\newblock \href {https://doi.org/10.46298/LMCS-18(2:19)2022}
  {\path{doi:10.46298/LMCS-18(2:19)2022}}.

\bibitem{Worrell00}
James Worrell.
\newblock Coinduction for recursive data types: partial orders, metric spaces
  and {$\Omega$}-categories.
\newblock In Horst Reichel, editor, {\em Coalgebraic Methods in Computer
  Science, {CMCS} 2000}, volume~33 of {\em ENTCS}, pages 337--356. Elsevier,
  2000.
\newblock \href {https://doi.org/10.1016/S1571-0661(05)80356-1}
  {\path{doi:10.1016/S1571-0661(05)80356-1}}.

\end{thebibliography}

\onecolumn
\appendix

\section{Omitted Proofs}

\subsection{Proof of \autoref{p:48}}

\begin{itemize}[wide]
	\item[]\ref{p:51} \(\Rightarrow\) \ref{p:53}.  Given
	      $r=g^\circ\cdot f$, the hypothesis of \ref{p:53} means that
	      $f(x_1)=g(y_1)=f(x_2)=g(y_2)$, so $x_1\mathrel{r} y_2$. %
	\item[]\ref{p:53} \(\Rightarrow\) \ref{p:54}.  The pullback of the
	      pushout is a relation $X\relto Y$ that relates $x\in X$ to $y\in Y$
	      iff~$x$ and~$y$ are equivalent under the equivalence relation on
	      $X+Y$ generated by~$r$. By \ref{p:53}, such elements $x,y$ are
	      already related by~$r$; that is, the pullback is~$r$. If
	      $r=\pi_2\cdot\pi_1^\circ$ as in \ref{p:54}, then~$R$ maps surjectively
	      onto~$r$, hence is a weak pullback.
	\item[]\ref{p:54} \(\Rightarrow\) \ref{p:51}.
	      Immediate, since every relation~$r$ can be written in the form
	      $r=\pi_2\cdot\pi_1^\circ$. \qed
\end{itemize}

\subsection{Proof of \autoref{p:910}}

Let $\ftF \colon \SET \to \SET$ be a functor.
It is clear that if $\ftF$ preserves 1/4-mono pullbacks, then it preserves inverse images.
To see that the converse statement holds, suppose that $\ftF$ preserves inverse images and consider a cospan $X \xto{f} B \xfrom{g} Y$ in $\SET$.
Then a pullback of the cospan $(f,g)$ can be obtained by pasting the following pullbacks where the bottom horizontal arrows are given by the image factorization of $f$.
\begin{center}
  \begin{tikzcd} %
    P & P'   & Y \\
    X & f[X] & B
    \ar[from=1-1, to=2-2, phantom, very near start, "\lrcorner"]
    \ar[from=1-1, to=1-2]
    \ar[from=1-1, to=2-1, "m"']
    \ar[from=1-3, to=2-3, "g"]
    \ar[from=2-1, to=2-3, bend right, "f"']
    \ar[from=2-1, to=2-2, twoheadrightarrow]
    \ar[from=2-2, to=2-3, tail]
    \ar[from=1-2, to=2-3, phantom, very near start, "\lrcorner"]
    \ar[from=1-2, to=2-2, "m'"]
    \ar[from=1-2, to=1-3, tail]
  \end{tikzcd}
\end{center}
Moreover, if $m$ is injective, i.e., if we have a 1/4-mono pullback
(square), then by the way pullbacks are formed in~$\SET$, $m'$ is also
injective. Therefore, in this case, the pullback of $(f,g)$ is
preserved because $\ftF$ preserves inverse images. \qed

\subsection{Details of \autoref{expl:functors}}

\begin{itemize}[wide]
	\item To see that the functor $(-)^3_2$ preserves 1/4-mono pullbacks,  consider a pullback
	\begin{displaymath}
		\begin{tikzcd} %
			P        \ar[from=1-1, to=2-2, phantom, very near start, "\lrcorner"] & Y \\
			X & Z.
			\ar[from=1-1, to=1-2, "p_2"]
			\ar[from=1-1, to=2-1, tail, "p_1"']
			\ar[from=1-2, to=2-2, "g"]
			\ar[from=2-1, to=2-2, "f"']
		\end{tikzcd}
	\end{displaymath}
	Note that as $p_1 \colon P \rightarrowtail X$ is injective, for every $x \in X$ such that $f(x) \in g[Y]$ there is one and only one element $y \in Y$ such that $f(x) = g(y)$.
	Now, let $x = (x_1,x_2,x_3)$ and $y = (y_1,y_2,y_3)$ be elements of $X^3_2$ and $Y^3_2$, respectively, such that $(f(x_1),f(x_2),f(x_3)) = (g(y_1),g(y_2),g(y_3))$.
	Then, from the fact that $x$ consists of at most two elements of $X$, we conclude that $((x_1,y_1), (x_2,y_2), (x_3,y_3))$ consists of at most two elements of $P$ and it is clear that projecting this element to $X$ and $Y$ yields $x$ and $y$, respectively.

	\item To see that the monotone neighbourhood functor and the clique functor weakly preserve 4/4-epi pullbacks, consider a pullback
\begin{equation}
	\label{p:951}
		\begin{tikzcd} %
			P        \ar[from=1-1, to=2-2, phantom, very near start, "\lrcorner"] & Y \\
			X & Z.
			\ar[from=1-1, to=1-2, twoheadrightarrow, "p_2"]
			\ar[from=1-1, to=2-1, twoheadrightarrow, "p_1"']
			\ar[from=1-2, to=2-2, twoheadrightarrow, "g"]
			\ar[from=2-1, to=2-2, twoheadrightarrow, "f"']
		\end{tikzcd}
\end{equation}
Suppose that there are $\calA \in \calM X$ and $\calB \in \calM Y$ such that $\calM f (\calA) = \calM f (\calB)$.
We have to show that there is $\calE \in \calM P$ such that $\calM p_1 (\calE) = \calA$ and $\calM p_2 (\calE) = \calB$.
Put
\[
	\calE = \upc \{ p_1^\circ[A] \mid A \in \calA\} \cup \upc \{ p_2^\circ[B] \mid B \in \calB\},
\]
where, given $r \colon U \relto V$ and $C \subseteq U$, $r[C]$ denotes the relational image.
Then, it is clear that $\calE$ is monotone.
Moreover, as $p_1$ is surjective, $\calA \subseteq \calM p_1 (\calE)$, since for every $A \in \calA$, $A = (p_1 \cdot p_1^\circ)[A] \in \ftM p_1(\calE)$.
On the other hand, for every set $B \subseteq Y$, $(p_1 \cdot p_2^\circ)[B] = (f^\circ \cdot g) [B] = f^\circ [g [B]]$, since \ref{p:951} is a (weak) pullback.
In particular, for every $B \in \calB$ we obtain $(p_1 \cdot p_2^\circ)[B] \in \calA$, since $g[B] \in \calM g(\calB) = \calM f(\calA)$.
Thus, as $\calA$ is monotone, for every set $C \subseteq P$ such that $p_2^\circ[B] \subseteq C$, for some $B \in \calB$, $p_1[C] \in \calA$.
This means that, $\calA \supseteq \calM p_1 (\calE)$, and, hence, $\calM p_1 (\calE) = \calA$.
By analogous reasoning we obtain $\calM p_2(\calE) = \calB$.
Therefore, the monotone neighborhood functor weakly preserves 4/4-epi pullbacks.
Now, suppose that $\calA$ and $\calB$ are cliques.
Let $A \in \calA$ and $B \in \calB$, we show that $p_1^\circ[A] \cap p_2^\circ[B] \neq \varnothing$, the other cases follow from this one or from the fact that $\calA$ and $\calB$ are cliques.
Note that, as $f[A], g[B] \in \calM f(\calA) = \calM g(\calB)$ which is a clique, we have $f[A] \cap g[B] \neq \varnothing$.
Hence, there is $a \in A$ and $b \in B$ such that $f(a) = g(b)$ and, therefore, $(a,b) \in p_1^\circ[A] \cap p_2^\circ[B]$, by definition of pullback.
	\item To see that the functor $\ftF \colon \SET \to \SET$ that maps a set~$X$ to the free semigroup over~$X$ quotiented by the equation $x x x = x x$ does not preserve 4/4-epi pullbacks weakly, consider the following pullback, where $2=\{a,b\}$.
\begin{displaymath}
		\begin{tikzcd} %
			P        \ar[from=1-1, to=2-2, phantom, very near start, "\lrcorner"] & 2 \\
			2 & 1
			\ar[from=1-1, to=1-2, twoheadrightarrow, "p_2"]
			\ar[from=1-1, to=2-1, twoheadrightarrow, "p_1"']
			\ar[from=1-2, to=2-2, twoheadrightarrow, "!_2"]
			\ar[from=2-1, to=2-2, twoheadrightarrow, "!_2"']
		\end{tikzcd}
\end{displaymath}
Then,  $\ftF!_2(aba) = \ftF!_2(ab)$ but there is no element $p$ in $\ftF P$ such that $\ftF p_1(p) = aba$ and $\ftF p_2(p) = ab$ since the words have different length and do not contain the pattern $xx$.
\end{itemize}

\subsection{Proof of \autoref{p:62}}

	Note that \(g^\circ \cdot f \geq g'^\circ \cdot f'\) means precisely that for all \(x \in X\) and \(y \in Y\), if \(h \cdot f(x) = f'(x) = g'(y) = h \cdot g(y)\), then \(f(x) = g(y)\).
	\begin{enumerate}[wide]
		\item Suppose that \(g^\circ \cdot f \geq g'^\circ \cdot f'\). Let \(a' \in f'[X]
		      \cap g'[Y]\), that is, we have \(x \in X\) and \(y \in Y\) such that \(h \cdot
		      f(x) = f'(x) = g'(y) = h \cdot g'(y)\). Then \(a: = f(x) = g(y) \in f[X] \cap
		      g[Y]\) and \(h(a) = a'\), so \(h'\) is surjective. On the other hand, let
		      \(a_1, a_2 \in f[X] \cap g[Y]\) such that \(h'(a_1) = h'(a_2)\). Then there are
		      \(x_1, x_2 \in X\) and \(y_1,y_2 \in Y\) such that \(a_1 = f(x_1) = g(y_1)\),
		      \(a_2 = f(x_2) = g(y_2)\) and, hence, in particular we obtain \(h \cdot f(x_1)
		      = h \cdot g(y_2)\). Therefore, \(a_1 = f(x_1) = g(y_2) = a_2\).
		\item Let \(a\in A\) with \(h(a) \in f'[X] \cap g'[Y]\). Then there are \(x \in X\)
		      and \(y \in Y\) such that \(h(a) = h \cdot f(x) = h \cdot g(y)\). Moreover,
		      since the cospan \(X\xrightarrow{\;f\;}A\xleftarrow{\;g\;}Y\)
		      is epi, w.l.o.g, there is \(x' \in X\) such that \(f(x') = a\). Hence, \(h
		      \cdot f (x') = h(a) = h \cdot g(y)\). Therefore, as \(g^\circ \cdot f \geq
		      g'^\circ \cdot f'\), \(a = f(x') = g(y) \in f[X] \cap g[Y]\).
		\item Let \(x \in X\) and \(y \in Y\) such that \(f'(x)=h \cdot f(x) = h \cdot
		      g(y)=g'(y)\). Since \eqref{d:eq:1} is a pullback,
		      it follows that $f(x),g(y)\in f[X] \cap g[Y]$, and since~$h'$ is injective, we
		      obtain $f(x)=g(y)$. \qed
	\end{enumerate}

\subsection{Proof of \autoref{p:64}}

	\begin{itemize}[wide]
		\item
		      \ref{p:49} \(\Rightarrow\) \ref{p:52}.
		      Let \(g^\circ \cdot f \colon X \relto Y\) be a difunctional relation determined by a cospan
		      $X\xto{f} A\xfrom{g} Y$. Consider the pushout $X\xto{p_1} O\xfrom{p_2} Y$
		      of the pullback of the cospan $X\xto{f} A\xfrom{g} Y$.
		      Note that every cospan that determines the relation \(g^\circ \cdot f\) gives rise to the
		      same pushout.  Furthermore, by \autoref{p:48},
		      \(p_2^\circ \cdot p_1 = g^\circ \cdot f\), and, hence, to show
		      that the claim holds it suffices to show
		      \((\ftF p_2)^\circ \cdot \ftF p_1 = (\ftF g)^\circ \cdot \ftF
		      f\).
		      By the universal property of $(p_1,p_2)$ as a pushout, we
		      have $h$ such that $h\cdot p_1=f$ and $h\cdot p_2=g$.  The
		      inequality
		      \((\ftF p_2)^\circ \cdot \ftF p_1 \leq (\ftF g)^\circ \cdot \ftF
		      f\) is then immediate from $\ftF h\cdot \ftF p_1=\ftF f$ and
		      $\ftF h\cdot \ftF p_2=\ftF g$.
		      To see the inequality \((\ftF p_2)^\circ \cdot \ftF p_1 \geq (\ftF g)^\circ \cdot \ftF f\) we consider first the case where \(g^\circ \cdot f\) is non-empty.
		      Note that, as \((p_1,p_2)\) is an epicocone, by \autoref{p:62} the inequality \(p_2^\circ \cdot p_1 \geq g^\circ \cdot f\) entails that we have the following pullback square
		      \begin{center}
			      \begin{tikzcd}[column sep=small]
				      p_1[X] \cap p_2[Y] & f[X] \cap g[Y] \\
				      O                  & A,
				      \ar[from=1-1, to=1-2, "\simeq"]
				      \ar[from=1-1, to=2-1, tail, "i_O"']
				      \ar[from=1-2, to=2-2, tail, "i_A"]
				      \ar[from=2-1, to=2-2, "h"']
				      \ar[from=1-1, to=2-2, very near start, "\lrcorner", phantom]
			      \end{tikzcd}
		      \end{center}
		      where $i_O$ and $i_A$ are the corresponding inclusions into $O$ and $A$, respectively.
		      Hence, since \(\ftF\) is 1/4-iso preserving, its image under \(\ftF\) is also a pullback.
		      Moreover, as $\SET$-functors preserve epimorphisms, by applying $\ftF$ to the commutative diagram
		      \begin{center}
			      \begin{tikzcd}[column sep=small]
			      	P &                    & Y \\
			      	  & p_1[X] \cap p_2[Y] & p_2[Y] \\
				      X & p_1[X]             & O
				      \ar[from=1-1, to=2-2, very near start, "\lrcorner", phantom]
				      \ar[from=1-1, to=2-2, dotted, twoheadrightarrow, bend left=20]
				      \ar[from=1-1, to=1-3]
				      \ar[from=1-1, to=3-1]
				      \ar[from=2-2, to=3-2, tail]
				      \ar[from=2-2, to=2-3, tail]
				      \ar[from=2-2, to=3-3, "i_O"', bend left=15]
				      \ar[from=2-3, to=3-3, tail]
				      \ar[from=3-2, to=3-3, tail]
				      \ar[from=2-2, to=3-3, very near start, "\lrcorner", phantom]
				      \ar[from=1-3, to=2-3, twoheadrightarrow]
				      \ar[from=1-3, to=3-3, bend left=45, "p_2"]
				      \ar[from=3-1, to=3-2, twoheadrightarrow]
				      \ar[from=3-1, to=3-3, bend right, "p_1"']
			      \end{tikzcd}
		      \end{center}
		      we conclude that $\ftF i_O \colon \ftF(p_1[X] \cap p_2[Y]) \to \ftF O$ corestricts to $\ftF p_1[\ftF X] \cap \ftF p_2[\ftF Y]$.
		      And, as $p_1[X] \cap p_2[Y]$ is non-empty since the relation $p_2 \cdot p_1^\circ$ is non-empty, $\ftF i_O$ is a monomorphism because every $\SET$-functor preserves monomorphisms with non-empty domain.
		      On the other hand, as $p_1[X] \cap p_2[Y]$ in non-empty and every $\SET$-functor preserves non-empty intersections \cite{Trn69}, we have $\ftF(p_1[X] \cap p_2[Y]) \simeq \ftF(p_1[X]) \cap \ftF(p_2[Y]) \simeq \ftF p_1[\ftF X] \cap \ftF p_2[\ftF Y]$, with the second isomorphism holding due to the fact that for every function $q \colon X \to Y$ with non-empty domain, the sets $\ftF(q[X])$ and $\ftF q[\ftF X]$ are isomorphic because each of them is the codomain of an epimorphism and the domain of a monomorphism of an epi-mono factorizations of $\ftF q$.
		      Hence, $\ftF i_O \colon \ftF (p_1[X] \cap p_2[Y]) \to \ftF O$ corestricts to an isomorphism $\ftF i_O \colon \ftF (p_1[X] \cap p_2[Y]) \to \ftF p_1[\ftF X] \cap \ftF p_2[\ftF Y]$.
		      And, by analogous reasoning for the morphism $\ftF i_A \colon \ftF (f[X] \cap g[Y]) \to \ftF A$, we obtain the commutative diagram
		      \begin{center}
			      \begin{tikzcd}[column sep=small]
			      	\ftF (p_1[X] \cap p_2[Y])              & \ftF (f[X] \cap g[Y])            \\
			      	\ftF p_1[\ftF X] \cap \ftF p_2[\ftF Y] & \ftF f[\ftF X] \cap \ftF g[\ftF Y] \\
				      \ftF O                                 & \ftF A.
				      \ar[from=1-1, to=1-2, "\simeq"', "\ftF h'"]
				      \ar[from=1-1, to=2-1, "\simeq"]
			        \ar[from=1-1, to=3-1, tail, bend right=70, "\ftF i_O"', end anchor={[yshift=-2ex]north west}]
			        \ar[from=1-2, to=3-2, tail, bend left=70, "\ftF i_A", end anchor={[yshift=-2ex]north east}]
			        \ar[from=1-2, to=2-2, "\simeq"']
			        \ar[from=2-1, to=2-2, dotted, "\simeq"]
			        \ar[from=2-1, to=3-1, tail]
			        \ar[from=2-2, to=3-2, tail]
			        \ar[from=3-1, to=3-2, "\ftF h"']
			      \end{tikzcd}
		      \end{center}

		      Thus, as the outer square is a pullback, the square
		      \begin{center}
			      \begin{tikzcd}[column sep=small]
				      \ftF p_1 [\ftF X] \cap \ftF p_2 [\ftF Y] & \ftF f[\ftF X] \cap \ftF g[\ftF Y] \\
				      \ftF O                                   & \ftF A
				      \ar[from=1-1, to=1-2, "\simeq"]
				      \ar[from=1-1, to=2-1, tail]
				      \ar[from=1-2, to=2-2, tail]
				      \ar[from=2-1, to=2-2, "\ftF h"']
				      \ar[from=1-1, to=2-2, very near start, "\lrcorner", phantom]
			      \end{tikzcd}
		      \end{center}
		      is a pullback, where the top morphism is given by restricting $\ftF h$ to $\ftF p_1 [\ftF X] \cap \ftF p_2 [\ftF Y]$.
		      Therefore, from \autoref{p:62}(\ref{item:difun-incl}), \((\ftF p_2)^\circ \cdot \ftF p_1 \geq (\ftF g)^\circ \cdot \ftF f\).
		      Now suppose that \(g^\circ \cdot f\) is empty.
		      Consider the functions \(f_{+1} \colon X+1 \to A+1\) and \(g_{+1} \colon Y+1 \to A+1\) that are defined as \(f\) and \(g\) on \(X\) and \(Y\), respectively, and send the element added to \(X\) and \(Y\), respectively, to the element added to \(A\), and the functions \(p_{1,+1} \colon X+1 \to O+1\) and \(p_{2,+1} \colon Y+1 \to O+1\) defined analogously.
		      Then, \(p_{2,+1}^\circ \cdot p_{1,+1} = g_{+1}^\circ \cdot f_{+1}\) is non-empty.
		      Hence, by the previous argument, \((\ftF p_{2,+1})^\circ \cdot \ftF p_{1,+1} = (\ftF g_{+1})^\circ \cdot \ftF f_{+1}\).
		      Now, let \(\fx \in \ftF X\), \(\fy \in \ftF Y\) such that \(\ftF f(\fx) = \ftF g(\fy)\).
		      Then, as the diagram
		      \begin{displaymath}
			      \begin{tikzcd}%
				      X & X + 1 \\
				      A & A + 1 \\
				      Y & Y + 1
				      \ar[from=1-1, to=1-2, tail, "i_X"]
				      \ar[from=1-1, to=2-1, "f"']
				      \ar[from=1-2, to=2-2, "f_{+1}"]
				      \ar[from=2-1, to=2-2, tail]
				      \ar[from=3-1, to=2-1, "g"]
				      \ar[from=3-1, to=3-2, tail, "i_Y"']
				      \ar[from=3-2, to=2-2, "g_{+1}"']
			      \end{tikzcd}
		      \end{displaymath}
		      commutes (with the horizontal arrows denoting coprojections), \(\ftF p_{1,+1}(\ftF i_X(\fx)) = \ftF p_{2,+1} (\ftF i_Y(\fy))\).
		      Hence, as the diagram
		      \begin{displaymath}
			      \begin{tikzcd}%
				      X     & X + 1     \\
				      O     & O + 1 \\
				      Y     & Y + 1
				      \ar[from=1-1, to=1-2, tail, "i_X"]
				      \ar[from=1-1, to=2-1, "p_1"']
				      \ar[from=1-2, to=2-2, "p_{1,+1}"]
				      \ar[from=2-1, to=2-2, tail, "i_O"]
				      \ar[from=3-1, to=2-1, "p_2"]
				      \ar[from=3-1, to=3-2, tail, "i_Y"']
				      \ar[from=3-2, to=2-2, "p_{2,+1}"']
			      \end{tikzcd}
		      \end{displaymath}
		      commutes (with the horizontal arrows denoting coprojections), \(\ftF i_O (\ftF p_1(\fx)) = \ftF i_O (\ftF p_2 (\fy))\).
		      Therefore, as \(\ftF\) preserves monomorphisms \(\ftF p_1(\fx) = \ftF p_2 (\fy)\) which entails \((\ftF p_2)^\circ \cdot \ftF p_1 \geq (\ftF g)^\circ \cdot \ftF f\).
		\item \ref{p:52} \(\Leftrightarrow\) \ref{p:110} The implication \ref{p:110} \(\Rightarrow\) \ref{p:52} is trivial.
		      To show \ref{p:52} \(\Rightarrow\) \ref{p:110}, let  \(g^\circ \cdot f \colon X \relto Y\) and \(g'^\circ \cdot f' \colon X \relto Y\) be difunctional relations, given by cospans
		      $X\xto{f} O\xfrom{g} Y$ and $X\xto{f'} A\xfrom{g'} Y$ respectively, such that \(g^\circ \cdot f \leq g'^\circ \cdot f'\).
		      Since \(\ftF\) is well-defined on difunctional relations, by \autoref{p:48} we can assume that the cospan $X\xto{f} O\xfrom{g} Y$ is the pushout of its pullback
		      $X\xto{\pi_1} R\xfrom{\pi_2} Y$.
		      Then, the condition \(g^\circ \cdot f \leq g'^\circ \cdot f'\) entails \(f' \cdot \pi_1 = g' \cdot \pi_2\).
		      Hence, by the pushout property, there is a map \(h \colon O \to A\) such that \(f' =  h \cdot f\) and \(g'= h \cdot g\).
		      Therefore, \((\ftF g)^\circ \cdot \ftF f \leq (\ftF g)^\circ \cdot (\ftF h)^\circ \cdot \ftF h \cdot \ftF f = (\ftF g')^\circ \cdot \ftF g'\).

		\item \ref{p:52} \(\Rightarrow\) \ref{p:10}.
		      Consider a pullback square of the form
		      \begin{equation*}
			      \begin{tikzcd}%
				      P & Y \\
				      X & Z.
				      \ar[from=1-1, to=2-2, phantom, very near start, "\lrcorner"]
				      \ar[from=1-1, to=1-2, "\simeq"]
				      \ar[from=1-1, to=1-2, "j"']
				      \ar[from=1-1, to=2-1, "i"']
				      \ar[from=1-2, to=2-2, "l"]
				      \ar[from=2-1, to=2-2, "f"']
			      \end{tikzcd}
		      \end{equation*}
		      Then, \(i \cdot j^{-1} = i \cdot j^\circ = f^\circ \cdot l\), so \(1_X^\circ \cdot (i \cdot j^{-1}) = f^\circ \cdot l\).
		      Hence, as \(\ftF\) is well-defined on difunctional relations, \((\ftF 1_X)^\circ \cdot \ftF(i\ \cdot j^{-1}) = (\ftF f)^\circ \cdot \ftF l\).
		      Thus, \(\ftF i \cdot (\ftF j)^\circ = \ftF i \cdot (\ftF j)^{-1} = \ftF(i \cdot j^{-1}) = (\ftF f)^\circ \cdot \ftF l\), i.e, \((\ftF i, \ftF j)\) is a weak-pullback of \((\ftF f, \ftF l)\).
		\item \ref{p:10} \(\Rightarrow\) \ref{p:49} trivial.
		      \qed
	\end{itemize}

\subsection{Proof of \autoref{p:34}}

	We begin by showing that, given a relax extension \(\eR\) of \(\ftF\), \(\laxif{\eR} \colon \REL \to \REL\) is a lax extension of \(\ftF\).
	\begin{itemize}[wide]
		\item[\ref{p:61}] Trivial.
		\item[\ref{p:26}] Let \(r \colon X \relto Y\) and \(s \colon Y \relto X\) be relations.
		      Moreover, suppose that \(r_1,\ldots, r_m\) and \(s_1, \ldots, s_n\) are finite sequences of relations such that \(r_1 \cdot \ldots \cdot r_m \leq r\) and \(s_1 \cdot \ldots \cdot s_n \leq s\).
		      Then, \(s_1,\ldots,s_n,r_1,\ldots,r_m\) is a finite sequence of relations such that \(s_n \cdot \ldots \cdot s_1 \cdot r_m \cdot \ldots \cdot r_1 \leq s \cdot r\).
		      Therefore, as relational composition preserves suprema in each variable,
		      \begin{align*}
			      \laxif{\eR} s \cdot\laxif{\eR} r & = \bigvee_{\substack{s_1, \ldots, s_n:    \\ s_n \cdot \ldots \cdot s_1 \leq s}} \bigvee_{\substack{r_1, \ldots, r_m: \\ r_m \cdot \ldots \cdot r_1 \leq r}} \eR s_n \cdot \ldots \cdot \eR s_1 \cdot\eR r_m \cdot \ldots \cdot \eR r_1 \\*
			                                       & \leq \bigvee_{\substack{t_1, \ldots, t_k: \\ t_k \cdot \ldots \cdot t_1 \leq s \cdot r}} \eR t_k \cdot \ldots \cdot \eR t_1 = \laxif{\eR} (s \cdot r).
		      \end{align*}
		\item[\ref{p:0}] Trivial.
	\end{itemize}
	Now, it is clear that \(\laxif{(-)} \colon \ReLaxF \to \LaxF\) is a monotone
	map and that the laxification of a relax extension produces a relax extension
	that is greater than or equal to the starting one, and that, since lax extensions
	preserve composition laxly, equality is attained precisely when the starting
	relax extension is a lax extension. Furthermore, suppose that \(\eR\) preserves
	converses. Let \(r \colon X \relto Y\) be a relation. Then, since \(r_1,
	\ldots, r_n\) is a composable sequence of relations such that \(r_n \cdots r_1
	\leq r\) iff \(s_1 = r_n^\circ, \ldots, s_n = r_1^\circ\) is a composable
	sequence of relations such that \(s_n \cdots s_1 \leq r^\circ\) and \(\eR\)
	preserves converses, we obtain:
	\begin{align*}
		(\laxif{\eR}r)^\circ & = \bigvee_{\substack{r_1, \ldots, r_n: \\ r_n \cdot \ldots \cdot r_1 \leq r}} (\eR r_n \cdot \ldots \cdot \eR r_1)^\circ \\
		                     & = \bigvee_{\substack{r_1, \ldots, r_n: \\ r_n \cdot \ldots \cdot  r_1 \leq r}} \eR r_1^\circ \cdot \ldots \cdot \eR r_n^\circ \\
		                     & = \bigvee_{\substack{s_1, \ldots, s_n: \\ s_n \cdots \ldots \cdot s_1 \leq r^\circ}} \eR s_n \cdot \ldots \cdot \eR s_1\\
		                     & = \laxif{\eR} (r^\circ) \tag*{\qed}
	\end{align*}

\subsection{Proof of \autoref{p:605}}

	Let $r \colon X_0 \relto X_n$ be a relation, and let $r_1, \ldots, r_n$ be a composable sequence of relations such that $r_n \cdot \ldots \cdot r_1 \leq r$.
	For $i=1,\ldots,n$, let $X_{i-1} \xfrom{\pi_i} R_i \xto{\rho_i} X_i$
	be a span in $\SET$ such that $r_i = \rho_i \cdot \pi_i^\circ$ and let $X_0 \xfrom{\pi_r} R \xto{\rho_r} X_n$ be a span in $\SET$ such that $r = \rho_r \cdot \pi_r^\circ$.
	We construct a sequence $r_1', \ldots, r_n'$ with the desired properties as follows.
	If $n=1$, we just take $r'_1 = r$, otherwise, with $[f,g] \colon X + Y \to Z$ denoting the copairing of $f \colon X \to Z$ and $g \colon Y \to Z$,
	\begin{itemize}[wide]
		\item $r'_1$ is given by the span $X_0 \xfromlong{[\pi_1,\pi_r]} R_1 + R \xtolong{\rho_1 + 1_R} X_1 + R$;
		\item for $i=2,\ldots,n-1$, $r'_i$ is given by the span $X_{i-1} + R \xfromlong{\pi_i + 1_R} R_i + R \xtolong{\rho_i + 1_R} X_i + R$;
		\item $r_n \colon (X_{n-1} + R) \relto X_n$ is given by the span $X_{n-1} \xfromlong{\pi_n + 1_R} R_n + R \xtolong{[\rho_n,\rho_r]} X_n$.
	\end{itemize}
	Then it is clear that by construction we have $r_n' \cdot \ldots \cdot r_1'=r$ and, hence, the claim follows from the fact that the diagram below commutes
	\begin{center}
		\begin{tikzcd}[column sep=small]
			X_0 & R_1     & X_1     & \ldots & X_{n-1}     & R_n    & X_n, \\
				  & R_1 + R & X_1 + R & \ldots & X_{n-1} + R & R_n+ R &
			\ar[from=1-2, to=1-1, "\pi_1"']
			\ar[from=1-2, to=1-3, "\rho_1"]
			\ar[from=1-6, to=1-5, "\pi_n"']
			\ar[from=1-6, to=1-7, "\rho_n"]
			\ar[from=2-2, to=1-1, "{[\pi_1,\pi_r]}"]
			\ar[from=1-2, to=2-2, tail]
			\ar[from=2-2, to=2-3, "\rho_1 + 1_R"']
			\ar[from=1-3, to=2-3, tail]
			\ar[from=1-5, to=2-5, tail]
			\ar[from=2-6, to=2-5, "\pi_n + 1_R"]
			\ar[from=1-6, to=2-6, tail]
			\ar[from=2-6, to=1-7, "{[\rho_n,\rho_r]}"']
		\end{tikzcd}
	\end{center}
	where the vertical arrows denote the corresponding coprojections. \qed

\subsection{Proof of \autoref{p:300}}

	Note that for all relations $r_1 \colon X \relto Y$ and $r_2 \colon Y \relto X$ such that $r_1 = \rho_1 \cdot \pi_1^\circ$ and $r_2 = \rho_2 \cdot \pi_2^\circ$, for spans $X \xfrom{\pi_1} R_1 \xto{\rho_1} Y$ and $Y \xfrom{\pi_2} R_2 \xto{\rho_2} X$ in $\SET$, $r_2 \cdot r_1 \leq 1_X \iff \pi_2^\circ \cdot \rho_1 \leq \rho_2^\circ \cdot \pi_1$.
	Therefore, the equivalence between \ref{p:601} and \ref{p:602} follows from the fact that preserving 1/4-iso pullbacks is equivalent to being monotone on difunctional relations ( \autoref{p:64}).
	The equivalence between \ref{p:602} and \ref{p:603} is an immediate consequence of \autoref{p:605}.
\qed

\subsection{Proof of \autoref{p:610}}

Let $X_1 \xto{\pi_2'} P \xfrom{\rho_2'}$ be the pushout of the cospan $(p_1,p_2)$.
So, $\hat r_2 = \rho_2' \cdot \pi_2^\circ$ and, as $\ftF$ is monotone and $p_1 \cdot \pi_2' = p_2 \cdot \rho_2'$,  $\ftbF r_3 \cdot \ftbF r_2 \cdot \ftbF r_1 \leq \ftbF r_3 \cdot \ftbF \hat r_2 \cdot \ftbF r_1 \leq \ftbF r_3 \cdot (\ftF p_2)^\circ \cdot \ftF p_1 \cdot \ftbF r_1 = \ftbF r_3' \cdot \ftbF r_1'$.

\subsection{Proof of \autoref{p:19}}

For \(i = 1,2,3\), let \(\rho_{i} \cdot {\pi_{i}}^\circ\) be the canonical factorization of \(r_i\).
We will show that ``splitting the elements'' of the domain of \(r_2\) that do not belong to the codomain of \(r_1\) and the elements of the codomain of \(r_2\) that do not belong to the domain of \(r_3\) yields a sequence of relations with the desired properties.
W.l.o.g.\ assume that \(Y \cap R_2=\emptyset\) and \(Z \cap R_2 = \emptyset\), and consider \(Y' = X_1 \cup (R_2\setminus \pi_2^{-1}[\rho_1[R_1]])\), \(Z' = X_2 \cup (R_2 \setminus \rho_2^{-1}[\pi_3[R_3]])\), which are then disjoint unions.
Consider the functions \(f \colon Y' \to Y\) and \(g \colon Z' \to Z\) that act identically on $Y$ and $Z$ and as \(\pi_2\) and \(\rho_2\) on \(R_2 \setminus \pi_2^{-1}[\rho_1[R_1]]\) and \(R_2 \setminus \rho_2^{-1}[\pi_3[R_3]]\) respectively.
Moreover, let \(p \colon R_2 \to Y'\) be the function that sends \((y,z) \in R_2\) to \(y \in Y\) if \(y \in \cod(r_1)\) and acts identically otherwise, and let \(q \colon R_2 \to Z'\) be the function that sends \((y,z) \in R_2\) to \(z \in Z\) if \(z \in \dom(r_3)\) and acts identically otherwise.
Then, we have the following commutative diagram
	\begin{center}
		\begin{tikzcd}%
			X & R_1 & Y & Y'   & R_2 & Z'   & Z & R_3 & W \\
			X & R_1 &     & Y    & R_2 & Z    &     & R_3 & W,
			\ar[from=1-2, to=1-1, "\pi_1"']
			\ar[from=1-1, to=2-1, equal]
			\ar[from=1-2, to=1-3, "\rho_1"]
			\ar[from=1-2, to=1-4, bend left, "\rho_1'"]
			\ar[from=1-2, to=2-2, equal]
			\ar[from=1-3, to=1-4, tail]
			\ar[from=1-2, to=2-4, phantom, "\lrcorner", very near start]
			\ar[from=1-4, to=2-4, "f"]
			\ar[from=1-8, to=2-6, phantom, "\llcorner", very near start]
			\ar[from=1-5, to=1-4, "p"']
			\ar[from=1-5, to=2-5, equal]
			\ar[from=1-5, to=1-6, "q"]
			\ar[from=2-2, to=2-1, "\pi_1"]
			\ar[from=2-2, to=2-4, "\rho_1"']
			\ar[from=2-5, to=2-4, "\pi_2"]
			\ar[from=2-5, to=2-6, "\rho_2"']
			\ar[from=1-6, to=2-6, "g"]
			\ar[from=1-7, to=1-6, tail]
			\ar[from=1-8, to=2-8, equal]
			\ar[from=2-8, to=2-6, "\pi_3"]
			\ar[from=2-8, to=2-9, "\rho_3"']
			\ar[from=1-8, to=1-7, "\pi_3"']
			\ar[from=1-8, to=1-6, bend right, "\pi_3'"']
			\ar[from=1-8, to=1-9, "\rho_3"]
			\ar[from=1-9, to=2-9, equal]
		\end{tikzcd}
	\end{center}
	where the arrows $Y \rightarrowtail Y'$ and $Z \rightarrowtail Z'$ denote inclusions.
	Since the second and the fourth squares are pullbacks, we obtain equations
	\(\rho_1' = f^\circ \cdot \rho_1\) and \(\pi_3'^\circ = \pi_3^\circ \cdot g\).
	Hence,
	\begin{align*}
		\rho_3 \cdot \pi_3^\circ \cdot \rho_2 \cdot \pi_2^\circ \cdot \rho_1 \cdot \pi_1^\circ & = \rho_3 \cdot \pi_3^\circ \cdot g \cdot q \cdot p^\circ \cdot f^\circ \cdot \rho_1 \cdot \pi_1^\circ \\
		                                                                                       & = \rho_3 \cdot \pi_3'^\circ \cdot q \cdot p^\circ \cdot \rho_1' \cdot \pi_1^\circ.
	\end{align*}
	Moreover, as \(\ftF\) is 1/4-iso pullback preserving, by applying \(\ftF\) to
	the commutative diagram above and reasoning analogously, we have
	\begin{align*}
		\ftbF r_3 \cdot \ftbF r_2 \cdot \ftbF r_1 & = \ftF \rho_3 \cdot (\ftF \pi_3)^\circ \cdot \ftF \rho_2 \cdot (\ftF \pi_2)^\circ \cdot \ftF \rho_1 \cdot (\ftF \rho_1)^\circ \\
		                                          & = \ftF \rho_3 \cdot (\ftF \pi_3')^\circ \cdot \ftF q \cdot (\ftF p)^\circ \cdot \ftF \rho_1 \cdot (\ftF \pi_1)^\circ .
	\end{align*}
	Therefore, with \(r_1' = \rho_1' \cdot \pi_1^\circ\), \(r'_2 = q \cdot
	p^\circ\) and \(r_3' = \rho_3 \cdot \pi_3'\),
	\[
		\ftbF r_3 \cdot \ftbF r_2 \cdot \ftbF r_1 = \ftbF r_3' \cdot \ftbF r_2' \cdot \ftbF r_1.
	\]
	Note that as \(\cod(r_1') = \cod(r_1)\) and \(\dom(r_3') = \dom(r_3)\), by
	construction, for all \(y,y' \in X_1\) and \(z,z' \in X_2\):
	\begin{enumerate}
		\item if \(y \neq y'\), \(y \mathrel{r'_2} z\) and \(y'\mathrel{r'_2} z\), then \(z
		      \in \dom(r'_3)\), and
		\item if \(z \neq z'\), \(y \mathrel{r'_2} z\) and \(y \mathrel{r'_2} z'\), then \(y
		      \in \cod(r'_1)\).
		      \qed
	\end{enumerate}

\subsection{Proof of \autoref{p:301}}

	\begin{itemize}[wide]
		\item \ref{p:701} $\Rightarrow$ \ref{p:702}
		Let $r_1 \colon X \relto Y$, $r_2 \colon Y \relto Z$ and $r_3 \colon Z \relto X$ be relations such that $r_3 \cdot r_2 \cdot r_1 \leq 1_X$.
		By \autoref{p:19} we can assume w.l.o.g that for all \(y,y'\) in \(X_1\) and \(z,z' \in X_2\):
		\begin{enumerate}
			\item \label{p:500} if \(y \neq y'\), \(y \mathrel{r_2} z\) and \(y'\mathrel{r_2} z\), then \(z \in \dom(r_3)\), and
			\item \label{p:501} if \(z \neq z'\), \(y \mathrel{r_2} z\) and \(y \mathrel{r_2} z'\), then \(y \in \cod(r_1)\).
		\end{enumerate}

		Let \(\hat{r}_2 \colon X_1 \relto X_2\) denote the difunctional closure of \(r_2\).
		We claim that \(r_3 \cdot \hat{r}_2 \cdot r_1 \leq 1_X\).
		Let \(x\mathrel{r_1}y_0\mathrel{r_2} z_1\mathrel{r_2^\circ} y_1 \mathrel{r_2} z_2 \dots \mathrel{r_2^\circ} y_{n-1}\mathrel{r_2} z_n\mathrel{r_3}x'\), with \(n \ge 1\);
		we have to show that \(x=x'\).
		We assume w.l.og that for \(i=0,\dots,n-2\), \(y_i \neq y_{i+1}\), and for \(j=1,\dots,n-1\), \(z_j \neq z_{j+1}\).
		The reason is that whenever \(y_i = y_{i+1}\), correspondingly $z_j = z_{j+1}$, we can remove $y_i \mathrel{r_2} z_{i+1}$ and $z_{i+1} \mathrel{r_2^\circ} y_{i+1}$ from the chain of related elements and still obtain a chain of related elements from $x$ to $x'$.

		We proceed by induction on~\(n\), with the base case \(n=1\) holding because $r_3 \cdot r_2 \cdot r_1 \leq 1_X$, by hypothesis.
		Let \(x\mathrel{r_1}y_0\mathrel{r_2} z_1\mathrel{r_2^\circ} y_1 \mathrel{r_2} z_2 \dots \mathrel{r_2^\circ} y_{n-1}\mathrel{r_2} z_n\mathrel{r_3}x'\).
		Since $n \geq 2$, we have $y_0 \neq y_1$ by assumption, and $y_0 \mathrel{r_2} z_1$ and $y_1 \mathrel z_1$.
		Hence, by item \ref{p:500} and the fact that \(r_3 \cdot r_2 \cdot r_1 \leq 1_X\), we obtain \(z_1 \mathrel{r_3} x\).
		This entails, by analogous reasoning using item \ref{p:501}, that \(x \mathrel{r_1} y_1\).
	Thus, we obtain a chain of related elements \(x\mathrel{r_1} y_1 \mathrel{r_2} z_2 \dots \mathrel{r_2^\circ} y_{n-1}\mathrel{r_2} z_n\mathrel{r_3}x'\);
	so
	\(x=x'\) by the inductive hypothesis.
	Therefore, by \autoref{p:610} and \autoref{p:300} we obtain $\ftbF r_3 \cdot \ftbF r_2 \cdot \ftbF r_1  \leq 1_{\ftF X}$.
	\item \ref{p:702} $\Rightarrow$ \ref{p:703} Immediate consequence of \autoref{p:605}.
	\item \ref{p:703} $\Rightarrow$ \ref{p:701} Immediate consequence of \autoref{p:300} since $\ftbF$ is normal. \qed
	\end{itemize}

\subsection{Details of \autoref{p:920}}

	The symmetric-reflexive closure of $\sim$ is already transitive. The property
	that~$\ftF$ preserves 1/4-iso pullbacks can be equivalently formulated as follows: for every rank-1 term~$t$,
	$t[x/x_1,\ldots,x/x_n] \sim t[x/x_1,\ldots,x/x_n,z_1/y_1,\ldots,z_m/y_m]$ implies $t\sim t[z_1/y_1,\ldots,z_m/y_m]$
	where $\{x_1,\ldots,x_n,y_1,\ldots,y_m\}$ are all the variables of $t$, and
	$t$ being rank-1 means that it contains precisely one occurrence either of $f$ or of $g$.
	Preservation of 1/4-iso pullbacks then follows by case-by-case analysis.
	That $\ftbF r_4 \cdot \ftbF r_3 \cdot \ftbF r_2 \cdot \ftbF r_1 \not\leq 1_X$  follows from the fact that
	we can define $h(x,y) = f(x,x,x,y,y)$, $u(x,y) = g(x,x,y,y,y)$, and then $h(x,y)\nsim u(x,y)$,
	which easily follows by inspecting the above clauses that define $\sim$. \qed

\subsection{Proof of \autoref{p:600}}

Clearly, the condition is necessary due to \autoref{cor:least_flat} and to see that is also sufficient first we note that by \autoref{p:301} it suffices to consider composable sequences of four or more relations.
Now, let \(r_1, \ldots, r_n\) be a composable sequence of relations such that $n \geq 4$, \(r_n \cdot \ldots \cdot r_1 = 1_X\), and \(r_i = X_{i-1} \relto X_i\) for \(i=1,\ldots,n\) and \(X_0 = X_n = X\).
We  show that this sequence admits a Barr upper bound such that all relations other than the first and the last are total and surjective.
Then, the claim follows by \autoref{cor:least_flat}.

For simplicity of notation let us assume that \(0,1 \notin X_i\), for \(i=2, \ldots, n-1\), and let \(X'_i = X_i \cup \{0,1\}\) for \(i=2,\ldots,n-1\).
Consider the sequence of relations \(r_1', \ldots, r_n'\) defined as follows:
\begin{itemize}[wide]
	\item The elements related by $r_1' \colon X \relto X_1'$ and $r_n' \colon X_{n-1}' \relto X$ are precisely the ones related by $r_1$ and $r_n$, respectively.
	\item For $i=2,\ldots,n-1$,  $r'_i \colon X'_{i-1}\relto X_i'$ consist of the following pairs
		\begin{itemize}[wide]
			\item \((0, x)\) if \(x=0\) or \(x \in X_i\setminus\cod(r_i)\);
			\item \((x, 1)\) if \(x=1\) or \(x \in X_{i-1}\setminus\dom(r_i)\);
			\item \((x,y)\)  if \(x \mathrel{r_i} y\).
		\end{itemize}
\end{itemize}

Then, by construction, for $i=2, \ldots, n-1$, \(r'_i\) is total and surjective and, with $\rho_i \cdot \pi_i^\circ$ and $\rho'_i \cdot \pi_i'^\circ$ denoting the canonical factorizations of $r_i$ and $r_i'$, respectively, the following diagram where the vertical arrows denote inclusions commutes

\begin{center}
	\begin{tikzcd}[column sep=small]
		X & R'_1 & X'_1 & R_2' & X'_2 &\ldots &X'_{n-2} & R_{n-1}' & X'_{n-1} & R_n' & X \\
		X & R_1 & X_1 & R_2 & X_2 & \ldots & X_{n-2} & R_{n-1} & X_{n-1} & R_n & X
		\ar[from=1-2, to=1-1, "\pi_1'"', twoheadrightarrow]
		\ar[from=1-2, to=1-3, "\rho_1'"]
		\ar[from=2-1, to=1-1, equal]
		\ar[from=2-2, to=1-2, equal]
		\ar[from=2-3, to=1-3, tail]
		\ar[from=2-5, to=1-5, tail]
		\ar[from=1-4, to=1-3, "\pi_2'"', twoheadrightarrow]
		\ar[from=1-4, to=1-5, "\rho_2'", twoheadrightarrow  ]
		\ar[from=2-4, to=2-3, "\pi_2"]
		\ar[from=2-4, to=2-5, "\rho_2"']
		\ar[from=2-4, to=1-4, tail]
		\ar[from=2-7, to=1-7, tail]
		\ar[from=2-2, to=2-1, "\pi_1", twoheadrightarrow]
		\ar[from=2-2, to=2-3, "\rho_1"']
		\ar[from=2-10, to=2-9, "\pi_n"]
		\ar[from=2-10, to=2-11, "\rho_n"', twoheadrightarrow]
		\ar[from=2-11, to=1-11, equal]
		\ar[from=2-10, to=1-10, equal]
		\ar[from=1-10, to=1-9, "\pi_n'"']
		\ar[from=1-10, to=1-11, "\rho_n'", twoheadrightarrow]
		\ar[from=1-8, to=1-7, "\pi_{n-1}'"', twoheadrightarrow]
		\ar[from=1-8, to=1-9, "\rho_{n-1}'", twoheadrightarrow]
		\ar[from=2-8, to=2-7, "\pi_{n-1}"]
		\ar[from=2-8, to=2-9, "\rho_{n-1}"']
		\ar[from=2-8, to=1-8, tail]
		\ar[from=2-9, to=1-9, tail]
	\end{tikzcd}
\end{center}

This entails that $r_n \cdot \ldots \cdot r_1 \leq r'_n \cdot \ldots \cdot r'_1$ and by applying $\ftF$ to the diagram we conclude that $\ftbF r_n \cdot \ldots \cdot \ftbF r_1 \leq \ftbF r'_n \cdot \ldots \cdot \ftbF r'_1$.
To see that $r_n \cdot \ldots \cdot r_1 \geq r'_n \cdot \ldots \cdot r'_1$ first note that by construction for $i=2,\ldots,n-1$, $x \mathrel{(r_i' \cdot \ldots \cdot r_2')} 0$ iff $x=0$ and $1 \mathrel{(r_{n-1}' \cdot \ldots \cdot r_i')} x$ iff $x=1$.
Now, suppose that $x \mathrel{r_1'} x_1 \mathrel{r_2'} x_2 \ldots \mathrel{r_{n-1}'} x_{n-1} \mathrel{r_n'} y$, then, as $0 \notin \cod(r_1')$ and $1 \notin \dom(r_n')$, for $i=1,\ldots,n-1$, $x_i \in X_i$.
Therefore, $x \mathrel{(r_n \cdot \ldots \cdot r_1)} y$.
\qed

\subsection{Proof of \autoref{p:960}}

Note that a relation $s \colon Y \relto Z$ is total iff it can be factorized as $f \cdot e^\circ$, for some function $f \colon A \to Z$ and some surjective map $e \colon A \to Y$.
And, dually, a relation $r \colon X \relto Y$ is surjective iff it can be factorized as $h \cdot g^\circ$, for some function $g \colon A \to X$ and some surjective map $h \colon A \to Z$.

\subsection{Proof of \autoref{p:100}}

Note that a relation $s \colon Y \relto Z$ is a partial function iff it can be factorized as $f \cdot i^\circ$, for some function $f \colon A \to Z$ and some injective map $i \colon A \to Y$.
And, dually, a relation $r \colon X \relto Y$ is the converse of partial function iff it can be factorized as $j \cdot g^\circ$, for some function $A \to X$ and some injective map $j \colon A \to Z$.
\qed

\subsection{Proof of \autoref{p:401}}

\begin{lemma}
  \label{p:101}
  Let \(r_1, \ldots, r_n\) be a composable sequence of  relations.
  Consider the following composable sequences of  relations:
  \begin{enumerate}
    \item \(s_1,\ldots,s_n\) defined by \(s_n = r_n\), and  \(s_i = \domr{s_{i+1}} \cdot r_i\), for \(i = 1,\ldots,n-1\);
    \item \(t_1,\ldots,t_n\) defined by \(t_1 = s_1\), and \(t_i = s_i \cdot \codr{t_{i-1}}\), for \(i = 2,\ldots, n\).
  \end{enumerate}
  Then, \(t_n \cdots t_1 = s_n \cdots s_1 = r_n \cdots r_1\) and for every \(i=2,\ldots,n\), \(\cod(t_{i-1}) = \dom(t_i)\).
\end{lemma}
\begin{proof}
  Clearly, \(t_n \cdots t_1 = s_n \cdots s_1 = r_n \cdots r_1\). %
  Moreover, note that \(\dom(t_i) = \dom(s_i) \cap \cod(t_{i-1})\), and \(\cod(t_{i-1}) = s_{i-1}[\cod(t_{i-2})]\) if \(i>2\) and \(\cod(t_{i-1}) = \cod(s_1)\) if \(i=2\).
  Thus, we have $\cod(t_{i-1})\le \cod(s_{i-1})$ for $i=2,\ldots,n$, and since \(\cod(s_{i-1}) \leq \dom(s_i)\) for \(i = 2,\ldots,n\), it follows that \(\dom(t_i) = \dom(s_i) \cap \cod(t_{i-1}) = \cod(t_{i-1})\).
\end{proof}

  Let  \(r_1, \ldots, r_n\) be a composable sequence of  relations such that \(r_n \cdot \ldots \cdot r_1 = 1_X\), for some set $X$.
  We will show that there is a Barr upper bound $t_1', \ldots, t_n'$ of $r_1, \ldots, r_n$ consisting only of total and surjective relations.
  Then, the claim follows immediately by \autoref{cor:least_flat}.
  By \autoref{p:101}, we obtain sequences \(s_1, \ldots, s_n\) and \(t_1,\ldots, t_n\) such that \(t_n \cdots t_1 = 1_X\) and for \(i=2, \ldots, n\), \(\cod(t_{i-1}) = \dom(t_i)\).
  Therefore, since \(\ftF\) preserves 1/4-mono pullbacks, by \autoref{p:100},
  \begin{flalign*}
    &&  & \ftbF r_n \cdots \ftbF r_1 \\
    &&  &\; = \ftbF s_n \cdot \ftbF \domr{s_n} \cdot \ftbF r_{n-1} \cdots \ftbF r_1  &\\
    &&  &\; = \ftbF s_n \cdot \ftbF (\domr{s_n} \cdot r_{n-1}) \cdots \ftbF r_1  &\\
    &&  &\; = \ftbF s_n \cdot \ftbF s_{n-1} \cdots \ftbF r_1  &\\
    &&  &\; = \ftbF s_n \cdot \ftbF s_{n-1} \cdot \ftbF\domr{s_{n-1}} \cdot \ftbF r_{n-2} \cdots \ftbF r_1  &\\
    &&  &\;\quad\vdots\\
    &&  &\; =\ftbF s_n \cdots \ftbF s_1 \\
    &&  &\; = \ftbF s_n \cdots \ftbF s_2 \cdot \ftbF \codr{t_1} \cdot \ftbF t_1  &\\
    &&  &\; = \ftbF s_n \cdot \ldots \ftbF (s_2 \cdot \codr{t_1}) \cdot \ftbF t_1  &\\
    &&  &\; = \ftbF s_n \cdot \ldots \ftbF t_2 \cdot \ftbF t_1  &\\
    &&  &\; = \ftbF s_n \cdots \ftbF s_3 \cdot \ftbF \codr{t_2} \cdot \ftbF t_2 \cdot \ftbF t_1  &\\
    &&  &\;\quad\vdots\\
    &&  &\;= \ftbF t_n \cdots \ftbF t_1.
  \end{flalign*}
  Furthermore, given spans $X \xfrom{f} R \xto{g} Y$, $Y \xfrom{f'} R \xto{g'} Z$ and a monomorphism $i \colon Y \rightarrowtail A$, we have $g'\cdot f'^\circ \cdot i^\circ \cdot i \cdot g \cdot f^\circ = g'\cdot f'^\circ \cdot g \cdot f^\circ$.
  Therefore, as $\ftF$ preserves monomorphisms, we obtain a Barr upper bound $t_1', \ldots, t_n'$ of $r_1, \ldots, r_n$ consisting only of total and surjective relations by (co)restricting $t_i$ to its (co)domain
  (Note that as $t_n \cdot \ldots \cdot t_1 =1_X$, $t_1$ is total and $t_n$ is surjective).
\qed

 \subsection{Proof of \autoref{p:900}}

  First note that \(r_3\cdot r_2\leq r_1^{\circ}\) and \(r_2\cdot r_1\leq r_3^{\circ}\) since \(1_{X_1}\leq r_1\cdot r_1^{\circ}\) and \(1_{X_2}\leq r_3^{\circ}\cdot r_3\). Since $r_2\leq\hat r_2 = \bigvee_{n\in\nat} r_2\cdot (r_2^\circ\cdot r_2)^n$, we show that \(r_3\cdot r_2\cdot (r_2^\circ\cdot r_2)^n\cdot r_1\leq 1_X \), for all \(n\in\nat\). We proceed by induction on~$n$, with the base case $n=0$ being the hypothesis of the lemma. Assuming that the assertion is true for \(n\in\nat\), we calculate
  \begin{displaymath}
    r_3\cdot r_2\cdot (r_2^\circ\cdot r_2)^n\cdot r_2^{\circ}\cdot r_2\cdot r_1
    \leq r_1^{\circ}\cdot (r_2^\circ\cdot r_2)^n\cdot r_2^{\circ}\cdot r_3^{\circ}
    =(r_3\cdot r_2\cdot (r_2^\circ\cdot r_2)^n\cdot r_1)^{\circ}
    \leq 1_X^{\circ}=1_X.\qedhere
  \end{displaymath}
 \qed

\end{document}